\documentclass[orivec,runningheads,final]{llncs}
\usepackage[T1]{fontenc}
\usepackage{graphicx}
\usepackage[small]{caption}
\usepackage[switch]{lineno}
\usepackage{stmaryrd}
\usepackage{paralist}
\usepackage[tickmarkheight=5pt,obeyDraft]{todonotes}
\usepackage[font=small,skip=0pt]{caption}
\usepackage{subfig}
\usepackage[ruled,linesnumbered,noend]{algorithm2e}
\usepackage{hyperref}
\usepackage{enumitem}
\usepackage{bm}
\usepackage{amssymb}
\usepackage{mathtools}
\usepackage{tikz}
\usetikzlibrary{automata, positioning, arrows}
\usepackage{accents}
\usepackage{multirow}
\usepackage{siunitx}
\usepackage{pdflscape}
\usepackage{pgfplots}
\usepgfplotslibrary{colorbrewer}
\pgfplotsset{compat = 1.15, cycle list/Set1-8} 
\usetikzlibrary{pgfplots.statistics, pgfplots.colorbrewer} 
\usepackage{pgfplotstable}
\usepackage{filecontents}

\usepackage{marginnote}
\usepackage[]{cleveref}

\newtheorem{appendixLemmaInner}{Lemma}
\newenvironment{appendixLemma}[1]{%
  \renewcommand\theappendixLemmaInner{#1}%
  \appendixLemmaInner
}{\endappendixLemmaInner}

\newtheorem{appendixTheoremInner}{Theorem}
\newenvironment{appendixTheorem}[1]{%
  \renewcommand\theappendixTheoremInner{#1}%
  \appendixTheoremInner
}{\endappendixTheoremInner}

\makeatletter
\newcommand{\customlabel}[2]{%
\protected@write \@auxout {}{\string \newlabel {#1}{{#2}{}}}}
\makeatother

\newif\ifArxiv
 \Arxivtrue

\newif\ifRelease
 \Releasetrue

\ifRelease
\newcommand{\change}[1]{{#1}\xspace}
\newcommand{\changeNew}[1]{{#1}\xspace}
\else
\newcommand{\change}[1]{{\color{blue}#1}\xspace}
\newcommand{\changeNew}[1]{{\color{orange}#1}\xspace}
\fi

\newcommand{\mo}[1]{#1}
\newcommand{\bl}[1]{#1}
\newcommand{\co}[1]{#1}

\newcommand{\NL}{\textsc{NL}\xspace}

\newcommand{\nop}[1]{}

\newcommand{\A}{\mathcal{A}}

\newcommand{\E}{\mathcal{E}}
\newcommand{\I}{\mathcal{I}}

\renewcommand{\L}{\mathcal{L}}

\newcommand{\R}{\mathcal{R}}

\newcommand{\T}{\mathcal{T}}

\newcommand{\eval}[2]{\llbracket #1 \rrbracket^{#2}}
\newcommand{\node}[1]{\langle #1 \rangle}
\newcommand{\transclosure}[2]{#1 \sqsubseteq^*_\T #2}

\newcommand{\epsilonpath}[2]{\varepsilon^*\text{-path from }#1\text{ to }#2}

\newcommand{\ISA}{\sqsubseteq}
\newcommand{\AND}{\sqcap}

\newcommand{\some}{\exists}

\newcommand{\rolenames}{\mathbf{R}}
\newcommand{\conceptnames}{\mathbf{C}}
\newcommand{\indivs}{\mathbf{N}}
\newcommand{\allroles}{\overline{\rolenames}}

\newcommand{\translate}[1]{\mathtt{rpe}(#1,G_\T)}

\newcommand{\union}{+}

\newcommand{\restrictedQuery}{NCQ}

\newcommand{\cdg}{CDG\xspace}
\newcommand{\ontoLang}{\mathcal{ELH}i^\mathit{ql}}
\newcommand{\invin}{\mathbin{\widehat{\in}}}

\begin{document}
\title{Towards Practicable Algorithms for Rewriting Graph Queries beyond DL-Lite}
\titlerunning{Towards Practicable Rewriting of Graph Queries beyond DL-Lite}
%

\author{Bianca Löhnert\inst{1}
\and
Nikolaus Augsten\inst{1}
\and
Cem Okulmus\inst{2}
\and
Magdalena Ortiz\inst{3}
}

\authorrunning{B. Löhnert et al.}
%
\institute{$^1$University of Salzburg, Austria  $^2$Paderborn  University, Germany  $^3$TU Wien, Austria}
%
\maketitle              
\begin{abstract}
Despite the advantages that the virtual knowledge graph paradigm has brought to many application domains, state-of-the-art systems still do not support popular graph database management systems like Neo4j.
Their query rewriting algorithms focus on languages like conjunctive queries and their unions, which 
were developed for relational data and
are poorly suited for graph data. Moreover, they also limit the expressiveness of the ontology languages that admit rewritings, restricting them to 
\changeNew{the  
DL-Lite family that enjoys the so-called \emph{FO-rewritability}.}
In this paper, we propose a technique for rewriting a family of \emph{navigational queries} for a suitably tailored fragment of $\mathcal{ELHI}$.  
\mo{Leveraging navigational features in the target query language, we can include some widely-used axiom shapes not supported by DL-Lite.}
We implemented a proof-of-concept prototype that rewrites into Cypher queries, and tested it on a real-world cognitive neuroscience use case with promising results.

\end{abstract}
\section{Introduction}
The virtual knowledge graph (VKG) paradigm~\cite{Xiao2023}, also known as ontology-based data access, emerged from real-life scenarios \cite{Poggi2008} and has been deployed in  real-world systems such as Ontopic, Stardog and Mastro \cite{mastro,ontopic,stardog}.
In it, the data is described by a high-level conceptual layer that enables the user to formulate queries using  the familiar vocabulary of a domain ontology. 
In addition to facilitating query formulation, the knowledge in the ontology can be used to infer implicit information when querying possibly incomplete data. The key problem to address here is \emph{ontology-mediated query answering} (OMQA), where a query is to be evaluated not directly over a plain dataset, but over the consequences of the given dataset together with the knowledge in the ontology.
The VKG paradigm has seen widespread adoption and brought by significant savings in the data access and integration costs in a range of applications \cite{DBLP:journals/dint/XiaoDCC19}, but so far, it has been limited to relational database management systems (RDBMS). 
The predominant technique for ontology-mediated querying is \emph{query rewriting}, where an input query is transformed to incorporate the ontological knowledge so that the rewritten query can be evaluated over any input dataset using standard technologies (i.e., SQL queries and standard RDBMS), without further ontological reasoning, and give complete answers that take into account implicit facts \cite{Calvanese2007}. The so-called \emph{union of conjunctive queries}, which capture popular SQL fragments \cite{DBLP:books/aw/AbiteboulHV95} are used as source and target languages for query rewriting. However, this technique can only be deployed for ontology languages whose data complexity is not higher than that of \change{evaluating the FO-fragment of SQL.} 
The DL-Lite family of description logics (DLs) was introduced with this goal in mind, and has become the ontology language of choice for OMQA~\cite{Calvanese2007}. \\ \indent
While traditional RDBMS are not going away, recent years have seen huge adoption of new paradigms for storing and querying data, and bringing the VKG paradigm to them would open up opportunities in many application domains. 
Particularly popular are \emph{graph databases} 
that adopt the \emph{property graph} (PG) data model~\cite{Angles2017}. 
PGs comprise nodes and edges between them; both nodes and edges can be labelled, and have assigned key-value pairs. 
In addition to a new data model, graph databases 
also provide unique querying abilities 
not found in the relational setting. 
\bl{Very recently the International Organization for Standardization (ISO) released a new standard for graph query languages. This standard,} called GQL \cite{Francis2023}, captures and extends \emph{conjunctive two-way regular path queries} (C2RPQ), a query language for graphs widely studied in the OMQA literature. \\ \indent
Query languages for graph data like C2RPQ and GQL are characterized by the \emph{navigational features} that allow queries to traverse paths of arbitrary length that comply with some \emph{regular path expression}. This fundamentally recursive feature results in a higher data complexity than that of SQL, and graph query languages are typically \NL complete in data complexity (under the so-called \emph{walk semantics}). This higher complexity, however, also means that we are not bound to DL-Lite and can potentially consider richer \change{ontology} languages with \NL data complexity. To our knowledge, the only such DLs so far are the linear fragments of $\mathcal{ELH}$ and $\mathcal{ELHI}$ identified by Dimartino et al.~\cite{Dimartino2019}, which \bl{do not allow any form of conjunction.} 
That work also introduced a query rewriting procedure for instance queries and CQs based on finite state automata but, unfortunately, it seems that it was never implemented.
When it comes to navigational ontology mediated queries, there is a rich body of theoretical work that considers C2RPQs as the input query language~\cite{DBLP:conf/ijcai/BienvenuOS13,DBLP:conf/kr/BienvenuCOS14,DBLP:journals/jair/BienvenuOS15}, but the goal of those works was to understand the computational complexity and they do not provide implementable query rewriting algorithms. 
The only practical approach to rewriting navigational queries is the recent algorithm for DL-Lite \cite{DBLP:conf/dlog/DragovicO023}. 
\changeNew{We extend this algorithm to a fragment of $\mathcal{ELHI}$ with NL data complexity. 
Unsurprisingly, lifting the techniques from DL-Lite is far from trivial. 
Even without conjunctions, in  $\E\L$ we may need a path of unbounded length to witness the propagation of a single atom, instead of single data point as in DL-Lite; with conjunction we may need a tree of such paths. 
Our fragment $\ontoLang$ was tailored to bound the branching points of such trees, while still accommodating the conjunctions and inverse roles used in our use case ontology from the domain of neuroscience.}
Our main contributions are as follows.
\newline 
\noindent
\textbullet~ We explore the limits of rewriting navigational queries and show that, even for lightweight ontologies, rewritings may not exist if C2RPQs are considered as both source and target query language. \\
\noindent
\textbullet~ Inspired by this insight, we identify a subset of C2RPQs, termed  Navigational Conjunctive Queries (NCQs), which \mo{offers key features of graph query languages---like reachability using the Kleene star---without ruling out the possibility of rewriting into unions of C2RPQs.}\\
\noindent
\textbullet~ 
\mo{We also push the boundaries of expressivity in the ontology language. Leveraging the navigational features of our target graph query language, we can accomodate typical $\mathcal{EL}$ axioms like $\exists{r}.A \ISA B$. Inverse roles and concept conjunction are two popular constructors heavily used in real-life ontologies. Hence, even though they make reasoning hard for PTime and thus preclude rewritability, we allow them in a restricted form. We call the resulting logic $\ontoLang$ (\emph{quasi-linear $\mathcal{ELH}$ with restricted inverses}), and we provide an algorithm for standard reasoning using a graph structure.} \\
\noindent
\textbullet~ We propose a rewriting for OMQs that pair NCQs with $\ontoLang$ ontologies. First we \change{introduce} a technique for standard reasoning to rewrite atomic queries into C2RPQs, and then we combine this with the well-known Clipper rewriting~\cite{Eiter2012} \change{that we extended for NCQs}. 
\\
\noindent
\textbullet~ We present a proof-of-concept prototype of our technique and use it to evaluate queries over real-world data from 
the domain of cognitive neuroscience.

\medskip\noindent
The remainder of the paper is structured as follows.
In \Cref{sec:preliminaries} we present the necessary terminology, such as the ontology we focus on, PGs in the context of OMQA and C2RPQs.  In  \Cref{sec:tboxreasoning}, we present our supported ontology language and show how reasoning about concept subsumption in the ontology can be handled via a bespoke data structure. In \Cref{sec:rewritingAQs} we first focus on the case of rewriting atomic queries into C2RPQs that capture all their consequences from the ontology.
In \Cref{sec:rewriting} we show the limitations of rewriting C2RPQs when used as the input and intended output language for OMQA and present a restricted subset of C2RPQs that retains rewritability.  
In \Cref{sec:impl_and_exp} we report on our proof-of-concept prototype and evaluate queries over data from a real-world use case. We conclude and point to future work in \Cref{sec:conclusion}. 
\ifArxiv
All proofs are found in the appendix to improve readability.
\fi

\section{Preliminaries}
 \label{sec:preliminaries}
	
\textbf{Ontology language.}
\mo{We recall the definition of $\mathcal{ELHI}$, a well-known description logic subsuming the popular lightweight languages DL-Lite$_\R$ and $\mathcal{ELH}$;
in the next sections we restrict our attention to a fragment of it. For space reasons we omit disjointness axioms, but they can be easily incorporated.
We also recall the usual \emph{normal form} for $\mathcal{ELHI}$ TBoxes; the proof that every TBox can be normalized (in linear time) while preserving the semantics is standard.}


\changeNew{We assume disjoint, countably infinite sets $\conceptnames$, $\rolenames$,  $\indivs$ and \change{$\mathbf{K}$}  of \emph{concept names},  \emph{role names}, \emph{individuals} and \change{\emph{key names}}, respectively, as well as a \emph{concrete domain} $(\mathbf{D}, \mathcal{P}^D)$, with $\mathbf{D}$ a set of \emph{values} and $\mathcal{P}^D$ a set of binary predicates over $\mathbf{D}$.
For example, $(\mathbf{D}, \mathcal{P}^D)$ could contain the integers with the usual $=$, $\leq$, $\geq$ predicates.} 

	\begin{definition}[$\mathcal{ELHI}$] \label{def:descriptionLogic}
     The set of \emph{roles} is $\allroles=\rolenames \cup \{r^-\mid r\in \rolenames \}$, 
		and \emph{concepts} $C$ follow the syntax 
        $C := A\mid\top\mid \exists r.C\mid C \sqcap C$,         
		%
        where $A\in \conceptnames$ and $r\in \allroles$.
		A \emph{concept inclusion} (CI) has the form $C \ISA D$, where $C,D$ are concepts, and a \emph{role inclusion} (RI) the form $r\ISA s$, where $r,s$ are roles. 
        A TBox is a finite set of CIs and RIs, and it is said to be in \emph{normal form} if all inclusions take these forms:  
        { \begin{align*}	
		  &  \ A_1\sqcap\dots\sqcap A_n \ISA B
            &  & \ \exists r.A \ISA B 
            &  & \ A \ISA \exists r.B &
            &  \  r \ISA s   
		\end{align*}    
        }
         We use $\transclosure{}{}$ to denote the reflexive transitive closure of $\{(r,s)\mid r\ISA s\in \T\}$ and call $r$ a \emph{subrole of $s$ (in $\T$)} if $\transclosure{r}{s}$.
	\end{definition}

\changeNew{
The semantics are given via \emph{interpretations} of the form $\I=(\varDelta^\I,\cdot^\I)$, with $\varDelta^\I$ a non-empty set called the \emph{abstract domain}. $\cdot ^\I$ is the \emph{interpretation function}, which assigns to every $A\in \conceptnames$ a set $A^\I\subseteq \varDelta^\I$, to every $r\in \rolenames$ a relation $r^\I\subseteq \varDelta^\I\times\varDelta^\I$ and to every $k \in \mathbf{K}$ a relation $k^\I \subseteq (\varDelta^\I \cup ( \varDelta^\I \times \varDelta^\I  )  ) \times \mathbf{D}  $.
}
It is extended to concepts and CIs in the usual way, see \Cref{tab:elhi} in \Cref{app:Prelmin}. \mo{Modelhood and entailment are also standard.} 

\medskip\noindent
\textbf{Data model.}
    In this paper the data (or ABox) is given as 
    finite \emph{property graphs}. 
    	
	\begin{definition}
		A \emph{property graph} (PG) $\mathcal{A}$ has the form $(N,E,\mathsf{label},\mathsf{prop})$, where: 
		\begin{compactitem}    
			\item $N$ is a non-empty set of \emph{nodes}; 
			\item $E$ is the \change{set} of edges; it assigns to each role $r\in \rolenames$ a relation on $N \times N$ which we write in the form $r(n,n')$ and call it the set of \emph{$r$-labeled edges}; 
			\item $\mathsf{label}$ is a total function $N \rightarrow 2^{\conceptnames}$; 
			\item $\mathsf{prop}$ is a partial function $(N\cup E) \times \mathbf{K} \rightarrow \mathbf{D}$ mapping pairs $(u,k)$ with $u \in (N\cup E)$ and $k\in \mathbf{K}$ to a value in
            $\mathbf{D}$.
		\end{compactitem}
    If 
    $N \subseteq \indivs$ and it is finite, 
    we call it an \emph{ABox}. A pair of a TBox $\T$ and an ABox $\A$ is called a \emph{knowledge base}. \change{We say that $\A' = (N', E',\mathsf{label}',\mathsf{prop}')$
    is a \emph{subgraph} of $\A$
    if $N' \subseteq N$, $E' \subseteq E$,  $ \mathsf{label}'(n) = \mathsf{label}(n) $ for all $n \in N'$, and $\mathsf{prop}'(u,k) = \mathsf{prop}(u,k) $ for all $u \in N' \cup E'$, $k \in K'$.}

	\end{definition}

 \changeNew{Note that we allow 
key-value pairs only in the ABox, and that our definition of property graph allows only a single edge between each pair of nodes.}

Each interpretation can be seen as a property graph, and vice-versa. 

\begin{definition}
 \label{def:PropInterpret}
		For a property graph $\A=(N,E,\mathsf{label}, \mathsf{prop})$, define 
  $\I_{\A}=(N, \cdot ^\I)$ as follows:   
			$C^\I = \{n\in N	\mid C \in \mathsf{label}(n)\}$,
			$r^\I = \{(n,n') 	\mid r(n,n')\in E\}$ and
           \change{ $k^\I = \{ (u,d) \mid  d = \mathsf{prop}(u,k) \}$.}
		Conversely, 
  $\I=(\varDelta^\I,\cdot ^\I)$ induces a (possibly infinite) property graph $\mathit{PG}(\I) =(\varDelta^\I,E, \mathsf{label},\mathsf{prop})$, where 
        $E = \{r(n,n')\mid (n,n')\in r^\I \}$, 
        $\mathsf{label}(n) = \{C\in \conceptnames \mid n \in C^\I\}$ and 
        \change{$\mathsf{prop}(u,k) = \{d  \in \mathbf{D} \mid  (u,d) \in  k^\I  \}$}.
    $\I$ is a \emph{model} of an ABox $\A$, if $\A$ is a subgraph of $PG(\I)$. An ABox $\A$ is \emph{consistent} with a TBox $\T$ iff there is a model of $\A$ and $\T$.
\end{definition}


\medskip\noindent
	\textbf{Query Language.}
	We study \emph{conjunctive two-way regular path queries} (C2RPQs), the navigational query language for graphs that has received most attention in OMQA. We enhance C2RPQs by \emph{data tests} as in \cite{DBLP:conf/dlog/DragovicO023} to query for property values, assuming that the predicates in  \changeNew{$\mathcal{P}^D$ can be realized} in GQL and Cypher. 
    
    \begin{definition} \label{def:c2rpq}
    Let $T$ be a \emph{data test} defined as $T := k \odot v \mid T \land T \mid T\lor T \mid \lnot T$, where $k \in \mathbf{K} \text{, } v\in \mathbf{D} \text{ and } \changeNew{\odot \in \mathcal{P}^D}$ \text{a binary predicate}.
    A \emph{regular path expression} (RPE) $\pi$ is defined as follows, with $\pi^+ = \pi\pi^*$, $A\in \conceptnames$, $r\in \rolenames$ and $T$ a data test.
    \begin{align*}
         & \alpha := r \mid r^- \mid \node{A} \qquad \pi := \alpha  \mid \co{\{T\}} \mid \pi\pi \mid \pi\union\pi \mid \pi^* \mid \pi^+
    \end{align*}
    We assume a countably infinite set $\mathbf{V}$ of \emph{variables}, disjoint from $\conceptnames$, $\rolenames$, $\indivs$, $\mathbf{K}$ and $\mathbf{D}$ and define \emph{atoms} of the form 
$\pi(x,y)$ with $\pi$ an RPE and $x,y \in \mathbf{V}$. \change{Note that $\pi^*$ (resp. $\pi^+$) refers to the Kleene star (resp. Kleene plus), as known from formal language theory~\cite{DBLP:books/daglib/0086373}.}
A \emph{conjunctive two-way regular path query (with data tests)} (C2RPQ(d)) is a pair $(\varphi,\vec{x})$ where $\varphi$ is a conjunction of atoms $\pi_1(x_1,y_1) \land \cdots \land \pi_n(x_n,y_n)$ and the \emph{answer variables} $\vec{x} $ are a tuple of variables occurring in $\varphi$. 
We write $\mathit{vars}(q) \subseteq \mathbf{V} $ for the set of all variables occurring in atoms of a query $q = (\varphi,\vec{x})$, where $\vec{x} \subseteq \mathit{vars}(q)$.
A variable $x$ is \emph{unbound} in $q$ if $x \not \in \vec{x}$ and it occurs in exactly one atom of $q$. 
We may write $q(\vec{x}) := \varphi$ for a C2RPQ with answer variables $\vec{x}$.

	\end{definition}
	
\normalmarginpar

Atomic concept tests are always binary atoms $\node{A}(x,x)$, but we often shorten this to $A(x)$. 
\changeNew{
We may refer to queries of form $q(x) := A(x)$ and $q(x,y) := \pi(x,y)$ as \emph{atomic} (a.k.a.\,instance) queries. Note that for RPEs $\pi$ that do not use roles (e.g., combinations of concept tests), the $y$ in an atom $\pi(x,y)$ is irrelevant.} 

To illustrate our query language, we make use of our real-world use case from the domain of cognitive neuroscience, which is also the scope of our experiments in \Cref{sec:impl_and_exp}.
The query $q(x)$  retrieves all
datasets from an MRI with a certain specification and data on ambidextrous participants:\begin{align*}
q(x):=&\mathsf{\langle Dataset\rangle(x) }\land \mathsf{\{Manufacturer="SIEMENS"}\land \mathsf{MagnetFieldStrength} \ge 3\}(x)\\ 
&\land \mathsf{has}^*(x,y) \land \mathsf{\langle Participant\rangle}(y)\land \mathsf{\{Handedness="ambidextrous"\}}(y)
\end{align*}
\changeNew{We use here a concrete domain that contains strings and integers, with separate predicates, and 
data tests on the keys \textsf{Manufacturer}, \textsf{Handedness},
and \textsf{MagneticFieldStrength}. The Kleene star in the second atom is the distinctive navigational feature of graph query languages, absent from any FO-rewritable query language; it lets us explore paths of unbounded length across the property graph.}

\medskip 

Following the 
literature, we use the \emph{homomorphism} or \emph{walk} semantics~\cite{Angles2017} \bl{for evaluating regular path queries. 
We define the evaluation of a C2RPQ $q(\vec{x})$ in terms of a function $\eval{\cdot }{\A}$, defined Table \ref{tab:eval} in \Cref{app:Prelmin}}.
If a TBox is given,  we adopt the \emph{certain answer} semantics as usual in OMQA.  
\begin{definition}\label{def:certainAnswer}
Let $\A$ be a PG with nodes $N$, and let  $q(\vec{x})$ be a C2RPQ. 
A tuple $\vec{a}$ of nodes in $N$ is an \emph{answer to $q(\vec{x})$ over $\A$} if there exists a mapping $\mu\in\eval{\varphi}{\A}$ s.t. $\mu(\vec{x})=\vec{a}$. 
For ABox $\A$ and TBox $\T$, we call $\vec{a}$ a certain answer to $q(\vec{x})$ over $(\T,\A)$ if $\vec{a}$ is an answer to $q(\vec{x})$ in $\mathit{PG}(\I)$ for every model $\I$ of $\A$ and $\T$.  
\end{definition}

\section{Reasoning in $\ontoLang$}\label{sec:tboxreasoning}
We introduce $\ontoLang$, a fragment of $\mathcal{ELHI}$
that restricts both conjunction and the use of inverses. It is well known that the interaction of these two constructors with 'non-local' concepts that may propagate across the ABox causes P-hardness in data complexity~\cite{Calvanese2013}.

\begin{definition}[$\ontoLang$]
\label{def:normalform} 
Let $\T$ be an $\mathcal{ELHI}$ TBox in normal form.  A concept name $B$ is \emph{non-local} in $\T$ if there is an axiom of the form $\exists r.B\ISA \change{C}$, or $B\ISA A$ with $A$ being non-local; otherwise, $B$ is \emph{local}. 
$\mathcal{T}$ is an $\ontoLang$ TBox if: 
    \begin{compactitem} 
        \item every existential restriction $\some{r}.{C}$ has either $r \in \rolenames$ or $C = \top$,
        \item only role names are allowed in RIs, and
        \item every concept name $B$ in an axiom $A_1\sqcap\dots\sqcap A_n \ISA B$ is local in $\T$. 
    \end{compactitem}   
     For convenience, we recall the axiom shapes of $\ontoLang$:
        \customlabel{NF1}{NF1}%
        \customlabel{NF3}{NF2}%
        \customlabel{NF4}{NF3}%
        \customlabel{NF5}{NF4}%
        \customlabel{NF7}{NF5}%
        \customlabel{NF8}{NF6}%
	\begin{align*}	
		& (NF1) \ A_1\sqcap\dots\sqcap A_n \ISA B
            & (NF2) &\ \exists r.A \ISA B 
            & (NF3) & \ A \ISA \exists r.B \\
            & (NF4) \ \  r \ISA s 			
            & (NF5) & \ \exists r^-.\top \ISA B 
            & (NF6) & \ A \ISA \exists r^-.\top 
	\end{align*}  
\end{definition}
\begin{example}
    Consider the $\mathcal{ELHI}$ TBox $\T=\{\some r.B\ISA C, A\ISA B, {A_1\AND A_2\ISA A}$, $\some r.A_3\ISA A_1\}$. By \Cref{def:normalform} $B,A,A_3$ are non-local and since $A$ occurs on the right-hand side of an axiom in form of (\ref{NF1}) $\T$ does not fall into the fragment of $\ontoLang$. However, $\T$ without the axiom $\some r.B\ISA C$ is indeed an $\ontoLang$ TBox. 
\end{example}
$\ontoLang$ is related, but different from  the \emph{harmless fragment of linear $\mathcal{ELHI}$} studied in \cite{Dimartino2020}.
\changeNew{The latter allows no conjunction, and imposes that if $A\ISA \exists r_2.A_2$ and $\exists r_1.A_1\ISA A_2$ occur together then  neither $\transclosure{r_1}{r^-_2}$ nor $\transclosure{r_2}{r^-_1}$ holds.}  

\noindent 
To reason 
in $\ontoLang$ we introduce a structure called \emph{concept dependency graph}. 

\begin{definition}\label{def:CDG_nodes_and_edges}
    We consider graphs $G$ that contain three types of edges: 
    \emph{$\varepsilon$-pro\-pa\-ga\-ting edges} $\varepsilon(A,B)$,  
    \emph{$r$-propagating edges} $r(A,B)$, and \emph{$r$-existential edges}, $\exists{r}(A,B)$, where $r \in \allroles$.        
    A \emph{propagating path} $p$ from node $A_0$ (to node $A_n$)
    is a (possibly empty) sequence $A_0\rho_1A_1\dots \rho_iA_i\dots A_n$ where, for each $1 \leq i < n$, 
    $A_i$ is a node in $G$ and there is a ($\varepsilon$- or $r$-)propagating edge $\rho_i(A_i,A_{i+1})$
    ; note that $\rho_i \in \allroles \cup \{\varepsilon\}$.
    The \emph{role label} $r_1\dots r_k$ of $p$ is defined as the subsequence of $\rho_1 \dots \rho_n$ that contains only roles, that is, omitting all $\rho_i$ such that $\rho_i = \varepsilon$.  
    Then, we say that $A_0\rho_1A_1\dots \rho_n A_n$ \emph{propagates} the concept $\exists s_1\dots\exists s_k.D$, 
    if $\transclosure{r_j}{s_j}$ for each $1 \leq j \leq k$,  \mo{$A_n$ is a (possibly empty) conjunction of concept names}, 
    and $D$ is a (possibly complex) $\ontoLang$ concept $D$ such that $\models D \ISA A_n$.
\end{definition}

\begin{definition}\label{def:conceptdependencygraph}
The \textit{concept dependency graph} (CDG) of an $\ontoLang$ TBox $\T$, denoted $G_\T$, is the directed multigraph defined as follows.
\begin{itemize}[leftmargin=1em]
     \item The set of nodes $N_\T$ of $G_\T$ contains the top concept $\top$, all concept names occurring in $\T$, all concepts of the form $A_1\sqcap\dots\sqcap A_n$ that occur on the \mo{left-hand} side of axiom of shape (\ref{NF1}), and a node $W_{\exists r.B}$ for each concept 
     $\exists r.B$ 
     that occurs on the right-hand side of an axiom of shape (\ref{NF4}) or (\ref{NF8}) in $\T$.  

    \item The set of edges $E_\T$ of $G_\T$ is the smallest set that satisfies the following: 
    \begin{enumerate}[label=(\alph*)]
        \item \label{def:conceptdependencygraph:item:top} $\varepsilon(\top,A) \in E_\T$ for each $A \in N_\T$
        \item \label{def:conceptdependencygraph:item:NF1} $\varepsilon(A,B_1\sqcap\dots\sqcap B_n)\in E_\T$ for each $B_1\sqcap\dots\sqcap B_n\ISA A$ of shape (\ref{NF1}) in $\mathcal{T}$
        \item \label{def:conceptdependencygraph:item:closeconjunction} $\varepsilon(B_1\sqcap\dots\sqcap B_n, A)\in E_\T$ if in $G_\T$ there is a node $B_1\sqcap\cdots B_i\cdots\sqcap B_n$ and
        for each $B_i$, there is
        an $\varepsilon^*$-path from $B_i$ to $A$
        \item \label{def:conceptdependencygraph:item:NF3} $r(A,B) \in E_\T$ for each $\exists r.B\ISA A$  of shape (\ref{NF3}) or (\ref{NF7}) in $\mathcal{T}$
        \item \label{def:conceptdependencygraph:item:roleAndInverse} $r(A,\top)\in E_\T$ if in $G_\T$ there is a propagating path with role label $r$ from $A$ to $B$ and $\exists s^-.\top\ISA B$ of shape (\ref{NF7}) in $\T$ with $\transclosure{r}{s}$ 
        \item \label{def:conceptdependencygraph:item:NF4} $\exists r(A,{W_{\exists r.B}}) \in E_\T$ and $\varepsilon(B,W_{\exists r.B}) \in E_\T$ for each $A\ISA\exists r.B$ of shape (\ref{NF4}) or (\ref{NF8}) in $\mathcal{T}$
        \item \label{def:conceptdependencygraph:item:NF7_inv} $\varepsilon(A,W) \in E_\T$ if there is some $\exists r^-. \top \ISA A$  of shape (\ref{NF7}) in $\mathcal{T}$, an edge $\exists s(C,W) \in E_\T$, with $\transclosure{s}{r}$ and no $\epsilonpath{A}{W}$ already
        \item \label{def:conceptdependencygraph:item:NF7} $\varepsilon(A,W) \in E_\T$ if in $G_\T$ there is $\exists s^-(C,W)$ with $\transclosure{s}{r}$ and a propagating path with role label $r$ from $A$ to $C$ and no $\epsilonpath{A}{W}$ already
        \item \label{def:conceptdependencygraph:item:close} $\varepsilon(A,C)\in E_\T$ if in $G_\T$ there is $\exists s(C,W)$ for some role $r$ with $\transclosure{s}{r}$ 
        and some propagating path $p$ with role label $r$ from $A$ to $W$, and there is no $\varepsilon^*$-path from $A$ to $C$ already
    \end{enumerate} 
\end{itemize}
\end{definition}

\tikzset{
->, 
>=stealth', 
node distance=2cm, 
every state/.style={thick}, 
initial text=$ $, 
  initial distance=5mm,
  every initial by arrow/.style={*->},
  initial text={start}
}

\begin{example}[\cdg]\label{ex:conceptdependencygraph}
	For the TBox $\T=\{ A_2 \ISA A_1$, $\exists r.B_1 \ISA A_1$, $\exists r_3.B_1 \ISA B_3$, $A_3 \ISA A_2$, $\exists r_1.B_2 \ISA B_1$, $\exists r_2^-.\top\ISA A_3$, $s\ISA r_2$, $\exists r_2.B_3 \ISA B_2$, $B_1 \ISA \exists r_2.B_3\}$ the CDG $G_\T$
    is as follows. (For readability, we omit edges from Def. \ref{def:conceptdependencygraph} item \ref{def:conceptdependencygraph:item:top}). 
\begin{center}

\resizebox{!}{2cm}{%
\begin{tikzpicture}[node distance = 2cm]
\node[state] (b1) {${B_1}$};
\node[state, left of=b1] (b2) {${B_2}$};
\node[state, below of=b2] (b3) {${B_3}$};
\node[state, right of=b3,inner sep=1pt,minimum size=0pt,text width=6.5mm] (exr2B2) {\small $\ W_{\exists }$\\ \vspace{-6pt}$_{r_2.B_3}$};
\node[state, right of=b1] (a1) {${A_1}$};
\node[state, right of=a1] (a2) {${A_2}$};
\node[state, right of=a2] (a3) {${A_3}$};
\node[state, below of=a3](top){${\top}$};

\draw (b1) edge[above, bend left] node{$r_1$} (b2)
(b2) edge[below, bend left] node{$\varepsilon$} (b1)
(b2) edge[left] node{$r_2$} (b3)
(b3) edge[below] node{\phantom{a}$r_3$} (b1)
(b3) edge[above] node{$\varepsilon$} (exr2B2)
(b1) edge[right] node{$\exists r_2$} (exr2B2)
(a1) edge[above] node{$r$} (b1)
(a1) edge[above] node{$\varepsilon$} (a2)
(a2) edge[above] node{$\varepsilon$} (a3)
(a3) edge[above] node{$\varepsilon$} (exr2B2)
(a3) edge[left] node{$r_2^{-}$} (top);
\end{tikzpicture}
}
\end{center}

\end{example}
\change{Note that the auxiliary nodes $W_{\exists r.B}$ help propagate inferences involving existentially quantified objects, but they do not participate in propagating paths.} 

The propagating paths in the CDG witness entailments from $\T$.

\begin{example}
\label{ex:entailmentPropagatingPath}
    Let us continue with $\T$ and $G_\T$ from \Cref{ex:conceptdependencygraph} and the ABox $\A=\{r(n_0,n_1),r_1(n_1,n_2),r_1(n_2,n_3), r_2(n_3,n_4),B_3(n_4)\}$. 
    Consider the propagating path $A_1rB_1r_1B_2\varepsilon B_1r_1B_2r_2B_3$ of the \cdg~ that starts in $A_1$ and ends in $B_3$. This path propagates $\exists r \exists r_1 \exists r_1 \exists r_2.B_3$ from which we infer that $\T$ entails $\exists r\exists r_1\exists r_1\exists r_2.B_3\ISA A_1$ and therefore $(\T,\A)\models A_1(n_0)$.   
\end{example}

\noindent
To witness all relevant entailments, we also use \emph{witnessing sets}. 

\begin{definition}\label{def:witnessingSet}
    Consider a \cdg $G_\T$ with nodes $N_\T$. 
    We say that the set $W_A\subseteq N_\T \cap \conceptnames$ \emph{witnesses} $A$ (in $G_\T$) if (a) $A \in W_A$ or (b) $A$ has a propagating path to a conjunction all whose conjuncts are witnessed by $W_A$. 
    \changeNew{We say that $W_A$ is a \emph{proper} witnessing set for $A$ if no strict subset of it satisfies (a) and (b).} 
\end{definition}

\begin{example}\label{ex:witnessingSet}
    Consider the $\ontoLang$ TBox $\T=\{A_1\AND A_2\ISA A, B_1\AND B_2\ISA A_1, C_1\AND C_2\ISA A_2, \some r.C\ISA C_2\}$. The \cdg $G_\T$ has edges $E_\T=\{\varepsilon(A,A_1\AND A_2), \varepsilon(A_1,B_1\AND B_2), \varepsilon(A_2,C_1\AND C_2), r(C_2,C)\}$. Then the set of \change{proper} witnessing sets \change{for $A$} is 
    \[
    \mathbf{W}_A = \big \{    \{A\}, \{A_1,A_2\}, \{B_1,B_2,A_2\}, \{A_1,C_1,C_2\}, \{B_1,B_2,C_1,C_2\} \big \}
    \] 
    An entailment of the form $\T\models A_1\AND C_1\AND \exists r.C\ISA A$ is captured by a set $W_A=\{A_1,C_1,C_2\}$ witnessing $A$, where $C_2$ propagates $\exists r.C$. 
\end{example}

\noindent
With these notions in place, our CDG captures \emph{all} 
the relevant entailments. 

\begin{lemma}
    \label{lemma:SoundnessPropagatingPaths}
    Let $\T$ be an $\ontoLang$ TBox and let $G_\T$ be its \cdg.     If there is a proper witnessing set $W_A=\{A_1,\dots,A_n\}$ for $A$ \mo{and a set of concepts $\{C_1, \ldots, C_n\}$ (of the form $\some s_1\dots \some s_k.B$ with $k\ge 0$) such that, for each $1\le i \le n$, there is a path \co{in $G_\T$} from $A_i$ propagating $C_i$, then $\T \models C_1\AND\dots\AND C_n \ISA A$. }
\end{lemma}
\noindent
To show the converse \Cref{lemma:completenessPropagatingPaths}
we construct a model $\I_{\A,\T}$ from a given ABox $\A$ that makes true exactly the $\ontoLang$ inclusions witnessed by $G_\T$. 

\begin{lemma}
    \label{lemma:completenessPropagatingPaths}
    Let $\T$ be a $\ontoLang$TBox and let $G_\T$ be its \cdg. 
    If $\T \models C_1\AND\dots\AND C_n \ISA A$ where $C_1,\dots, C_n$ are concepts of the form $\some s_1\dots \some s_k.B$ with $k\ge 0$ and $A,B\in\conceptnames$, then there is a set $W_A = \{A_1, \cdots, A_n\}$ that witnesses $A$ \mo{and such that for each $1 \leq i \leq n$ there is a path \co{in $G_\T$} from $A_i$ that propagates $C_i$.}
\end{lemma}
The key to our rewriting of atomic queries is that
instance checking for a concept $A_0$ reduces to finding a set of ABox assertions in the given ABox that \bl{match a set of propagating paths that witnesses $A_0$}. 
\begin{lemma}
\label{lemma:relation}
Let $\T$ be an $\ontoLang$ TBox \co{ and let $G_\T$ be its \cdg.} 
For each ABox $\A$,  $A \in \conceptnames$ and individual $n_0$, we have 
$\T,\A \models A(n_0)$ iff there is a proper witnessing set $W_A=\{A_1,\dots,A_n\}$ for $A$ such that for each $A_j\in W_A$ there exists 
\textit{(i)} a (possibly empty) path $p_j$ \co{ in $G_\T$ }
with role label $(r^j_1\cdots r^j_k)$ that propagates a concept $\some s_1 \dots\some s_k.B_j$; and 
\textit{(ii)} a sequence of individuals $n^j_{1},\dots,n^j_{k}$ such that $B_j(n^j_{k}) \in \A$
and  $r^j_i(n^j_{i},n^j_{(i+1)}) \in \A$
for each $1 \leq i < k$.
\end{lemma}
\bl{Observe that in an $\ontoLang$ TBox, $A$ is local in all axioms $C_1\AND\dots\AND C_n\ISA A$,
and also every $B$ for which $\T\models A\ISA B$ holds must be local. Thus, $A$ can never participate on the left-hand  side of axioms of the form $\exists r.A\ISA C$, which could lead to entailments of the form like $\T\models \some r.(C_1\AND\dots\AND C_n)\ISA C$. In other words, all relevant entailments of an $\ontoLang$ TBox take the form $\T\models C_1\AND\dots\AND C_n\ISA A$.} Hence this Lemma is not hard to show using \Cref{lemma:SoundnessPropagatingPaths} and \Cref{lemma:completenessPropagatingPaths}. 
\ifArxiv
\newline The proofs of the above claims can be found in \Cref{app:ProofsSection StandardReasoning}.
\fi

\section{Rewriting Atomic Queries}\label{sec:rewritingAQs}

In this section we present a rewriting for instance queries under $\ontoLang$ TBoxes.
We do this relying on CDGs, which represent the witnessing sets and propagating paths that we have proved to provide all relevant entailments. 

We start by defining a non-deterministic finite automaton (NFA) that captures all paths in $G_\T$ that propagate a concept $C$. 

\begin{definition}[Path-Generating NFA]\label{def:pathGeneratingNFA}
    Let $G_\T$ be a CDG, and let $Q_\T$ contain the nodes in $G_\T$ that are $\ontoLang$ concepts. 
    For $B \in \conceptnames$, the \emph{$B$-path-generating automaton of $G_\T$}, termed $G_\T^{\mathbb{A}}(B)$, is the NFA 
    $\langle Q_\T, \Sigma,  \delta, q_0,F\rangle$ with: 
 \[
    \Sigma = \allroles \cup \conceptnames \cup  \{ \varepsilon \} \text{\phantom{space}}  
        q_0 = B \text{\phantom{space}} F = \{ \top \}  
\]
    \begin{flalign*}
    \delta(n_0,\alpha) &= 
        \begin{cases} 
         n_1 & \alpha(n_0,n_1) \in E_\T \\
          \top & \text{ if }  \alpha \in \conceptnames 
        \end{cases}
    \end{flalign*} 
\end{definition}

\begin{example}
In order to illustrate how the path-generating NFAs are constructed, we look at the CDG $G_\T$ given in \Cref{ex:conceptdependencygraph} and show its corresponding $B_1$-path-generating NFA $G_\T^{\mathbb{A}}(B_1)$ with the extracted RPE for $B_1$ underneath it. 
\begin{center}

\resizebox{!}{2.5cm}{%
\begin{tikzpicture}[node distance = 2cm]
\node[state, initial below] (b1) {${B_1}$};
\node[state, accepting,above of=b1] (omega) {$\top$};
\node[state,right of=b1](a1){${A_1}$};
\node[state,right of=a1](a2){${A_2}$};
\node[state, right of = a2](a3){${A_3}$};
\node[state, left of = b1](b2){${B_2}$};
\node[state, left of = b2](b3){${B_3}$};

\draw (a1) edge[above] node{r} (b1)
(a1) edge[above] node{$\varepsilon$} (a2)
(a2) edge[above] node{$\varepsilon$} (a3)
(b1) edge[right] node{$B_1$} (omega)
(a1) edge[right] node{$A_1$} (omega)
(a2) edge[right] node{$A_2$} (omega)
(a3) edge[right] node{\phantom{a}$A_3$} (omega)
(a3) edge[bend right, below] node{$r_2 ^{-} $} (omega)
(b2) edge[right] node{$B_2$} (omega)
(b3) edge[above] node{$B_3$} (omega)
(b2) edge[bend right, above] node{$\varepsilon$} (b1)
(b1) edge[bend right, above] node{$r_1$} (b2)
(b2) edge[above] node{$r_2$} (b3)
(b3) edge[bend right, below] node{$r_3$} (b1);
\end{tikzpicture}}
\end{center}
\begin{align*}
    \big ( \overbrace{( (r_1\union(r_1r_2r_3))^*B_1)}^{\text{paths ending in } B_1}
    \union \overbrace{( r_1^+ (r_2r_3r_1^+)^* B_2)}^{\text{paths ending in } B_2}
    \union \overbrace{( r_1^+ r_2 (r_3r_1^+r_2)^* B_3)}^{\text{paths ending in } B_3}      \big ) 
\end{align*}
\end{example}


We use $\translate{B}$ to denote the regular path expression constructed from $G_\T^{\mathbb{A}}(B)$ (using standard techniques);  note that here $\varepsilon$ is the usual NFA $\varepsilon$-transition. 
Clearly, $\translate{B}$ contains exactly the words of the form 
    $s_1, \dots, s_k A$ such that there is an $(s_1, \dots, s_k)$-propagating path from $B$ to $A$ in $G_\T$. 



We now provide a rewriting algorithm for queries of the form $q(x) := A(x)$. We first extract from $G_\T$ the witnessing sets for $A$, and then use   
$\translate{B}$ to consider all propagating paths for the concept names in each witnessing set.  
Using Lemma~\ref{lemma:relation}, it is not hard to show that Algorithm~\ref{algo:rewriteAQ} is a sound and complete query rewriting algorithm for atomic queries of the form $q(x) := A(x)$.

\begin{proposition}
Let $\T$ be an $\ontoLang$ TBox and $q(x) := A(x)$ a C2RPQ.
For every ABox $\A$ and tuple $\vec{a}$ of individuals, 
we have that $\vec{a}$ is a certain answer to $q(\vec{x})$ over $(\T,\A)$ if and only if $\vec{a}$ is an answer to $\mathtt{rewriteAtomQuery}(\vec{x})$ over $\A$.
\end{proposition}

\setlength{\textfloatsep}{1pt}
\begin{algorithm}[t]
\small
	\caption{Rewriting atomic concept queries 
	}
    \label{algo:rewriteAQ}
    
    \SetKwInOut{Input}{Input}
    \SetKwInOut{Output}{Output}
    \SetKwFunction{FrewriteAtomicConceptQuery}{$\mathtt{rewriteAtomQuery}$}
    \SetKwFunction{FrewriteRole}{$\mathtt{rewriteRole}$}
    \SetKwFunction{FWitness}{$\mathtt{witnessSets}$}
    \SetKwFunction{FrewriteConcept}{$\mathtt{rewrConcept}$}
    \SetKwProg{Fn}{function}{:}{}

    \Input{$A(x)$, $G_\T$, $\T$}
    \Output{UC2RPQ $Q$}
    \Fn{\FrewriteAtomicConceptQuery{$A(x),G_\T$}}{
        $Q:=\emptyset$ \\
        \ForEach{$\{B_1,\dots,B_n\}$ in \FWitness{$A,G_\T$}}{
            $q(x):=$ \FrewriteConcept{$B_1,G_\T$}$(x,y_1)\land \dots \land$ \FrewriteConcept{$B_n,G_\T$}$(x,y_n)$\\
            $Q:=Q\cup \{q(x)\}$\\
        }
        \Return $Q$
    }       

    \Fn{\FrewriteRole{$r$, $\T$}}{
        \label{line:UnionsRoles}  
        \Return $\bigcup r_i$ where $\transclosure{r_i}{r}$ \label{line:returnRewritingRole}
    }

    \Fn{\FrewriteConcept{$A$, $G_\T$}\label{func:rewriteRole}}{
        $\pi := \translate{A}$  \label{line:translate} \\
        \ForEach(\label{line:forPi}){$r \in \pi$}  {
            \label{line:superRoles}  $\pi := \pi[r \backslash$ \FrewriteRole{r,$\T$ }$]$
        }
        \Return $\pi$ \label{line:returnRewritingConcept}
    }
    \Fn{\FWitness{$A$, $G_\T$}\label{func:witnessSets}}{
        $\mathbf{W_A}:=\{\{A\}\}; \mathbf{W'_A}:=\emptyset$\\
        \While{$W_A\neq W'_A$}{
            $\mathbf{W'_A}:= \mathbf{W_A}$ \\
            \If{$B$ propagates $B_1\sqcap\dots\sqcap B_n$ in $G_\T$ and $W\in \mathbf{W}'_A$ with $B\in W$}{
                $\mathbf{W_A}:=\mathbf{W_A}\cup \big ( (W\setminus \{B\}) \cup \{B_1, \cdots B_n \} \big )$
            }
        }
        \Return $\mathbf{W_A}$
    }  
\end{algorithm}

\noindent
\textbf{Description of Algorithm \ref{algo:rewriteAQ}.} 
\changeNew{The function $\mathtt{witnessSets}$ takes an input a concept $A$ and a CDG, and iteratively computes all the proper witnesses by exhaustively replacing a concept with a conjunction whenever a corresponding path is found. 
The function $\mathtt{rewriteRole}$ simply checks for any other roles that imply $r$ w.r.t. $\T$ and produces the union over all of them. 
In case of function $\mathtt{rewrConcept}$, we first extract the RPE, as described above, 
and then use $\mathtt{rewriteRole}$ to 
replace each role in the produced RPE with the union of its subroles.
Finally, $\mathtt{rewriteAtomQuery}$ brings everything together: given an atomic query $A(x)$, we compute its witnessing sets and for each such set, rewrite each concept name into an RPE using function $\mathtt{rewrConcept}$. Note that as $A(x)$ is replaced by a set of atoms that may match to regular  paths, we introduce a fresh variable $y_i$ for each of them.  
We produce a C2RPQ by forming the conjunction of these atoms.
The output of $\mathtt{rewriteAtomQuery}$ is the union of these C2RPQs.}

\section{Rewriting Navigational Queries}\label{sec:rewriting}
In this section we provide an algorithm for rewriting navigational queries. \changeNew{Its pseudo-code description is given in \Cref{app:Algorithm2}.}
Unfortunately, it is in general not possible to rewrite C2RPQs into \change{unions of C2RPQs (UC2RPQs)} for any ontology language that allows $\T \models \exists r.\top \ISA \exists s.\top$ for role names $r$ and $s$, as we show in~\Cref{thm:c2rpq_rewriting}.
\ifArxiv
The proofs for this section can be found in~\Cref{app:ProofsSectionRewriteNQs}.
\fi
\begin{theorem}
    \label{thm:c2rpq_rewriting}
    There exist C2RPQs that cannot be rewritten into UC2RPQs w.r.t. TBoxes containing concepts of the form $\exists r.\top$ on both sides of CIs. 
    This holds already for C2RPQs with only one atom.
\end{theorem}


\noindent
This negative result motivates the need for a restricted query language that admits a rewriting into UC2RPQs under $\ontoLang$ TBoxes. 
We use a similar language to \cite{Dragovic2022}, called \emph{Navigational Conjunctive Queries} (\restrictedQuery{s}). 

\begin{definition}[Navigational Conjunctive Query]
		\label{def:NCQs}
		A \emph{Navigational Conjunctive Query} (\restrictedQuery) is a C2RPQ, with atoms restricted to the following forms:   
        {\small
        \[   \{T\}(x) \  \{T\}(x,y) \ \node{A_1} \union\cdots \union \node{A_n}(x,x)  \
            \big (\pi_1 \union\cdots \union \pi_n \big )(x,y) \            \big (\pi_1 \union\cdots \union \pi_n \big)^*(x,y) \]   
		}    
\noindent
 with $T$ a data test, $A_i \in \conceptnames$ 
 and  $\pi_i$ a \emph{restricted RPE} $\pi_i := r \mid r^- \mid r^* \mid (r^-)^*$.  
 	\end{definition}
While in \restrictedQuery{s} we cannot use concatenation, 
it can still be simulated outside the scope of Kleene stars as usual.

\newcommand{\restrictexp}[2]{#1|^\T{#2}}
\subsection{Rewriting Algorithm for \restrictedQuery{s}}

In this section we give the complete algorithm for rewriting \restrictedQuery{s} and a given $\ontoLang$ TBox into UC2RPQs. 
\bl{Let us recall that $\ontoLang$ does not make assertions about property values, meaning that data tests are only evaluated on individuals in the ABox. For the sake of simplicity, we can therefore omit atoms with data tests during rewriting, as they remain untouched in the query.}
For the purpose of rewriting \restrictedQuery{s}, we make use of the functions for rewriting single roles and concepts given in Section \ref{sec:rewritingAQs}. 
We thus introduce a new function, called \emph{clipping}, which makes use of axioms in the TBox to rewrite a given \restrictedQuery\ into a union of \restrictedQuery{s}, inspired by~\cite{Eiter2012}. Given an RPE $\pi$ and a role $r$, we construct a new expression $\restrictexp{\pi}{r}$ in the following way.
\begin{definition}\label{def:UnionOfStars}
    Let $\pi$ be an expression of the form $\pi_1 \union \cdots \union \pi_n$, and let $r$ be a role.
    Then the \emph{$r$-restriction of $\pi$}, written as $\restrictexp{\pi}{r}$, is the union of all those $\pi_j$ s.t. there exists some $s$ with $\transclosure{r}{s}$, where  $s$ is in $\pi_j$. 
\end{definition}
Informally, $\restrictexp{\pi}{r}$ matches only those paths of $\pi$ that contain the role $r$ or a super-role of $r$. For a set of atoms $\mathbb{A}$ and atom $\pi$ 
we use $\pi(x,y) \invin \mathbf{A}$ to mean  $\pi(x,y) \in \mathbf{A}$  or $\pi^-(y,x) \in \mathbf{A}$. Note that $\pi^-(x,y)$ is built from $\pi(y,x)$, where $r^- \in \pi^-$ iff $r \in \pi$. 
    
\begin{definition}[Clipping Function] \label{def:clipping}
Given a query $q(\vec{x})$, we select a set $Y$ of variables s.t. $Y \cap \vec{x} = \emptyset$, and 
 a CI $A\ISA \exists r.B \in \T$. 
Then, do the following:  
\begin{enumerate}[label={(D\arabic*)},leftmargin=2.5em]
    \item \label{def:clipping:variables} Pick any $y \in Y$, and replace each $y'\in Y$ by $y$ everywhere in $q(\vec{x})$.  
    \item \label{def:clipping:conditions} 
    Every atom  $\alpha$ where $y$ occurs needs to satisfy one or more of the following:
    \begin{enumerate}[label={(\Alph*)}]
    	\item\label{clipRule:A} $\alpha$ contains a star, and if it contains a variable different from $y$, it is an unbound variable. 
        \item\label{clipRule:B} 
        $\alpha$ contains a concept name $C$ with $ \T \vDash B \ISA C$,  or $ \T \vDash \exists r^{-}\change{.\top} \ISA C$, and if it contains a variable different from $y$, it is an unbound variable.
        \item\label{clipRule:C}  $\alpha$ is  of the form $\pi(x,y)$ or $\pi^-(y,x)$ where $x \neq y$ and $\pi$ contains some $s$ with $\transclosure{r}{s}$.
    \end{enumerate}
    \item\label{def:clipping:roleAtomSets} Let $\mathbf{C}_y = \{ \pi(x,y) \mid \pi(x,y) \text{ satisfies {(C)}}\} \cup \{ \pi^-(x,y) \mid \pi(y,x) \text{ satisfies {(C)}} \}$. Define $\mathbf{C}_y^*$ as the set of atoms $\alpha$ in $\mathbf{C}_y$ that contain a role $s$ with $\transclosure{r}{s}$ occurring in $\alpha$ in the scope of a star. Let $\mathbf{C}^1_y = \mathbf{C}_y \setminus \mathbf{C}^*_y$, and let $X$ be the set of all variables different from $y$ that occur in $\mathbf{C}_y^1$.
    \item\label{def:clipping:drop} Drop from $q(\vec{x})$ every atom satisfying {(A)} or {(B)}, and every atom $\alpha \invin \mathbf{C}^1_y$. 
    \item\label{def:clipping:replaceVar} Replace each $x \in X$ by $y$, everywhere in $q(\vec{x})$. 
    \item\label{def:clipping:replaceStarAtom} Replace each atom $\pi(x,y) \invin q$ such that $\pi(x,y) \in \mathbf{C}_y^*$ by $\restrictexp{\pi}{r}(x,y)$. 
    \item\label{def:clipping:addConcept} Add $A(y)$ to $q(\vec{x})$. 
  \end{enumerate}
\end{definition}

\noindent
The clipping function relies on a well-known property of $\mathcal{ELHI}$:  for each $\T$ and $\A$ there is a universal, tree-shaped model $\I^u$ that can be used for answering all C2RPQs. Intuitively, we consider each possible set $Y$ of variables that may be mapped to the same 'anonymous' object $d$ of maximal depth in $\I^u$. 
Each such $d$ is triggered by one existential axiom $A \ISA \exists r.B$, i.e., if $d$ was added to $\I^u$, we know that it has a parent $d_p$ that is $A$, and $d$ was introduced as an $r$-child  of $d_p$ to satisfy $A \ISA \exists r.B$. Using this, we can modify the query to require that we map a variable to the object $d_p$ that is $A$, and drop from the query all (parts of) atoms that are already guaranteed by the existence of $d$. 
We show in the correctness proofs that each application of the function results in a rewritten query with 
the same answers, but whose mappings have a strictly lower depth (as at least one variable now is mapped to $d_p$ instead of its child $d$).
By repeated application, we
obtain a query with all variables are mapped to ABox individuals.

\begin{example}
    Let us consider the TBox $\T=\{A \ISA \exists r.B\}$ and the query $q(x_1):=(t^*\union r^*)(x_1,x_2)\land s^*(x_2,x_3)\land B(x_3)\land r^-(x_2,x_4)\land C(x_4)\land t^*(x_4,x_5)$. In \cref{line:variableSet} of \Cref{algo:rewriteC2RPQ} we iterate over all subsets $Y$ of $\mathit{vars}(q)$. Suppose that $Y=\{x_2,x_3\}$, then after applying \ref{def:clipping:variables} we get the query $q(x_1):=(t^*\union r^*)(x_1,x_2)\land s^*(x_2,x_2)\land B(x_2)\land r^-(x_2,x_4)\land C(x_4)\land t^*(x_4,x_5)$. Each atom of the query fulfills one of the conditions in \cref{def:clipping:conditions}. In \ref{def:clipping:roleAtomSets} we get the sets $\mathbf{C}^*_{x_2}=\{(t^*\union r^*)(x_1,x_2)\}$ and $\mathbf{C}^1_{x_2}=\{r^-(x_2,x_4)\}$, note that $\mathbf{C}_{x_2}$ is the union of these two sets. Then, in the next step \ref{def:clipping:drop} we drop from $q$ every atom that satisfies \ref{clipRule:A}, \ref{clipRule:B}, and the set $\mathbf{C}^1_{x_2}$. After replacing the variables in $X=\{x_4\}$ by $x_2$, we obtain the query $q(x_1):=(t^*\union r^*)(x_1,x_2)\land C(x_2)$. In step \ref{def:clipping:replaceStarAtom} we replace the atom $(t^*\union r^*)(x_1,x_2)$ by $r^*(x_1,x_2)$. Observe that keeping $t^*$ in the query makes the rewriting become not sound, since we either query for a path with exclusive $r$ or $t$ labels.    Finally, in \ref{def:clipping:addConcept} we add $A(x_2)$ to the query and return $q(x_2):=r^*(x_1,x_2)\land C(x_2)\land A(x_2)$.
\end{example}
In order to reduce the number of redundant queries in the output, we extend \Cref{algo:rewriteC2RPQ} by a simple containment check that we call \emph{structural subsumption} (see function $\mathtt{add^\subseteq}$). Intuitively, a query $q$ is structurally subsumed by another query $q'$ w.r.t. some TBox $\T$ (if each atom of $q$ is more specific than an atom of $q'$); this guarantees that  
the set of answers of $q'$ always contains all the answers of $q$. 
This allows us to drop some queries and obtain a smaller rewriting.  

\begin{definition}
    \label{def:queryContainment}
    Let $\T$ be a TBox and $q,q'$ be NCQs, we say that $q$ is \emph{structurally subsumed} by $q'$ w.r.t. $\T$, or short $q\subseteq_\T q'$, if for each atom $\beta_1\union\dots\union\beta_j\union\dots\union\beta_m(x,y)\in q'$ there exists an atom $\alpha_1\union\dots\union\alpha_i\union\dots\union\alpha_n (x,y)\in q$, such that for each $\alpha_i$ there is a $\beta_j$ such that $\T\vDash\alpha_i \ISA \beta_j$ and $\alpha_i,\beta_j\in \conceptnames\cup\allroles$.
\end{definition}

\begin{example}
    Consider a TBox $\T=\{r\ISA s, A_1\ISA B_1, A_2\ISA B_2\}$ and the queries $q_1(x):=C(x),r(x,y)\land A_1\union A_2(y)$ and $q_2(x):=s(x,y)\land B_1\union B_2\union B_3(y)$. Then, by \Cref{def:queryContainment} $q_1$ is subsumed by $q_2$ w.r.t. $\T$. 
    Let's consider the ABox $\A_1=\{r(a,b), C(a),A_1(b)\}$. Observe that $a$ is an answer to $q_1(x)$ as well as $q_2(x)$ over $(\T,\A)$. However, this is not always the case for the opposite direction as we show by the ABox $\A_2=\{r(a,b),B_2(b)\}$. Here $a$ is indeed an answer to $q_2(x)$ over $(\T,\A_2)$, but not $q_1$. Hence, $q_1$ is subsumed by $q_2$
    w.r.t. $\T$, but not vice-versa. 
\end{example}

\begin{lemma}
    \label{lemma:queryContainment}
    Let $\T$ be a TBox, $\A$ be an ABox, and $q_1(\vec{x}),q_2(\vec{x})$ be two NCQs, such that $q_1\subseteq_\T q_2$. Then, $\vec{a}$ is a certain answer to $(\T,\A, q_2(\vec{x}))$ if $\vec{a}$ is a certain answer to $(\T,\A, q_1(\vec{x}))$. 
\end{lemma}


\noindent
\bl{\textbf{Algorithm to Rewrite NCQs.}} We now formally state the definition of the function \texttt{rewriteNCQ}.
\bl{As input we assume an \restrictedQuery\ $q(\vec{x})$ and an $\ontoLang$ TBox~$\T$. The function first iterates over all subsets of variables and axioms in $\T$ of the form (\ref{NF4}) and (\ref{NF8}), and applies the clipping function exhaustively, producing a union of NCQs $Q$.
\changeNew{
Next the function loops over the queries $q' \in Q$ and for each concept $A_i$ in $q'$, we compute its witnessing sets and between lines 9 to 12, we produce a conjunction using $\mathtt{rewrConcept}$, similar to \Cref{algo:rewriteAQ}. For roles $r$ occurring inside $q'$, between lines 13 and 14, we similarly produce a new query, replacing each occurrence of $r$ with the output of $\mathtt{rewriteRole}$. 
In order to reduce the number of redundant queries in $Q$ we check $q'$ against structural query subsumption over $Q$, and either remove all queries in $Q$ that are contained in $q'$ before adding it, or drop $q'$ if $q'$ itself is already structurally subsumed by some other element of $Q$. 
The result $Q$ of this rewriting is a union of C2RPQs.
}

}

\begin{theorem}
    \label{thm:correctnessNCQRewriting}
    Let $\T$ be an $\ontoLang$ TBox and let $q(\vec{x})$ be a \restrictedQuery. 
    For every ABox $\A$ and tuple $\vec{a}$ of individuals, we have that $\vec{a}$ is a certain answer to $q(\vec{x})$ over $(\T,\A)$ if and only if $\vec{a}$ is an answer to $\mathtt{rewriteNCQ}(q(\vec{x}),\T)$ over $\A$.
\end{theorem}

\noindent In the proofs we show that one step of the $\mathtt{clipping}$ function is sound and complete. 
In \co{function \texttt{rewriteNCQ}}
we exhaustively apply the clipping function, which means that we have a query that can be evaluated over the plain ABox without the anonymous part. 
\co{The correctness for the second part of function \texttt{rewriteNCQ}, where we substitute unions with conjunctions, is given by \Cref{lemma:SoundnessPropagatingPaths} and \Cref{lemma:completenessPropagatingPaths} and the exhaustive replacement of all reachable conjunctions  and rewriting of roles and atoms.}
By \Cref{thm:correctnessNCQRewriting} NCQ answering for $\ontoLang$ reduces to C2RPQ query evaluation, which is in \NL in data complexity (see \cite{Barcelo2013}). This is worst-case optimal. 

\begin{theorem}
    \label{thm:terminationNCQ}
    Let $\T$ be an $\ontoLang$ TBox and $q$ a NCQ. Then, the algorithm $\mathtt{rewriteNCQ}(\T,q)$ terminates. 
\end{theorem}
\section{Implementation and Experiments} 
\label{sec:impl_and_exp}

We implemented a proof-of-concept prototype that, 
given an $\ontoLang$ TBox (in OWL syntax), rewrites \restrictedQuery{s} into UC2RPQs and translates them into Cypher, a declarative query language for the Neo4j property graph database. The Cypher query is then evaluated over real-world data stored in Neo4j. The Java source code of the prototype is publicly available~\cite{owl2cypher}.
\smallskip
{\newline \noindent
\textbf{Setup.}}
We execute the experiments on a virtual cluster node running Rocky Linux 8.10 with an AMD EPYC 7513 32-Core CPU clocked at 2.60 GHz and 400 GB RAM; with Neo4j 5.18.1 running on the same machine.
\smallskip
{ \newline \noindent
\textbf{Ontology.}}
As TBox (OWL ontology) we use the Cognitive Task Ontology (COGITO)~\cite{COGITO}, which integrates concepts of the Cognitive Atlas \cite{Poldrack2011} with the Hierarchical Event Descriptors (HED) \cite{Robbins2021}. This ontology includes about \num[round-precision=0]{4700} concepts and \num[round-precision=0]{9200} axioms, all of them expressible in $\ontoLang$: \num[round-precision=0]{122} of the axioms contain conjunction (NF1) and existential quantifiers on the right (NF3). For example, the axiom $\mathsf{ReadingTask}\ISA (\exists \mathsf{has.Read}\sqcap \exists\mathsf{has.Lang\text{-}item})$  defines a reading task by referring to the HEDs $\mathsf{Read}$ and $\mathsf{Lang\text{-}item}$ (where the conjunction is just a shortcut for two axioms in normal form). For all axioms, COGITO also includes the converse axiom, e.g., $(\exists \mathsf{has.Read}\sqcap \exists\mathsf{has.Lang\text{-}item})\ISA \mathsf{ReadingTask}$. 
\smallskip
{  \newline \noindent
\textbf{Data.}} The prototype rewrites an \restrictedQuery~into a Cypher query assuming that concepts in the ontology correspond to node labels in the database, and roles correspond to relationships (i.e., edge labels).  
For the experiments we choose a dataset from the domain of cognitive neuroscience \cite{Ravenschlag2023a}. This dataset---stored in our Neo4j database---consists of \num[round-precision=0]{396741} nodes and \num[round-precision=0]{2870405} relationships. It contains meta-information about fMRI data from OpenNeuro~\cite{OpenNeuroMRI}.
\begin{table}[t]
\centering
    \caption{Properties of rewritten queries, rewriting and evaluation time in Neo4j.}
     \label{tab:experiments}
     \subfloat[Queries grouped by type  \label{tab:experiments:GroupByStructure}]{
        \begin{tabular}{c c | c c | c c | c}
        \hline
        \multicolumn{2}{c |}{\textbf{Group} } & \multicolumn{2}{c |}{\textbf{Rewritten Queries (Avg.)}} & \multicolumn{2}{c |}{\textbf{Runtime [s]}} & \textbf{\#Timeouts}\\ 
        {type} & {\#queries} & {\#answers} & {\#atoms} & {rewriting} & {evaluation} & {(600s)} \\
		\hline
		G1 & \num[round-precision=0]{114} & \tablenum{0,35} & \tablenum{27,13} & \tablenum{0,05335} & \tablenum{2,87753} & {0} \\
		G2 & \num[round-precision=0]{1041} & \tablenum{1,21} & \tablenum{5,00} & \tablenum{0,03565} & \tablenum{2,38956} & {6} \\
		G3 & \num[round-precision=0]{2060} & \tablenum{0,45} & \tablenum{52,46} & \tablenum{1,00672} & \tablenum{54,40565} & {27} \\
		G4 & \num[round-precision=0]{1041} & \tablenum{371,65} & \tablenum{2,79} & \tablenum{0,01520} & \tablenum{0,93819} & {0} \\ 
            G5 & \num[round-precision=0]{114} & \tablenum{1,04}& \tablenum{20,21} & \tablenum{0,02678} &  \tablenum{2,34997} & {0} \\
        \hline
        Total & \num[round-precision=0]{4370} & \tablenum{89,74} & \tablenum{27,70} & \tablenum{0,48617} & \tablenum{26,43589} & {33} \\
		\hline
	\end{tabular} 
    }
    
    \subfloat[Queries grouped by size (number of C2RPQs in union resulting from rewriting)\label{tab:experiments:GroupByQuerySize}]{

        \begin{tabular}{c c | c c | c c | c}
		\hline
    	\multicolumn{2}{c |}{\textbf{Group} } & \multicolumn{2}{c |}{\textbf{Rewritten Queries (Avg.)}} & \multicolumn{2}{c |}{\textbf{Runtime [s]}} & \textbf{\#Timeouts}\\ 
        {size} & {\#queries} & {\#answers} & {\#atoms} & {rewriting} & {evaluation} & {(600s)} \\
    		\hline
    		{1}-{10} & \num[round-precision=0]{3785} & \tablenum{101,20} & \tablenum{10,60} & \tablenum{0,16722} &  \tablenum{14,35577} & {5} \\
    		{11}-{20}  & \num[round-precision=0]{302} & \tablenum{12,23} & \tablenum{83,62} & \tablenum{1,45607} & \tablenum{80,90783} & {0} \\
    		{21}-{30} & \num[round-precision=0]{129} & \tablenum{1,63} & \tablenum{134,63} & \tablenum{2,73019} & \tablenum{127,83398} & {6} \\
    		{30}+ & \num[round-precision=0]{154} & \tablenum{21,11} & \tablenum{289,81} & \tablenum{5,30956} & \tablenum{153,25702} & {22} \\ 
            \hline
            Total & \num[round-precision=0]{4370} & \tablenum{89,74}& \tablenum{27,70} & \tablenum{0,48617} & \tablenum{26,43589} & {33} \\
		\hline
	\end{tabular} 
    }
\end{table}
\smallskip
{\newline \noindent
\textbf{Queries.} }One use case of COGITO is to query for fMRI data containing a specific set of HED concepts (e.g. \textsf{Lang-Item}, \textsf{Read}), even if the data has only annotations for cognitive task concepts (e.g. \textsf{ReadingTask}), or vice versa. \bl{
Our goal is to evaluate the effects of different input queries on our rewriting approach, which is not affected by the presence of data tests. Therefore, we do not include data tests in our queries. 
We} generated a total of \num[round-precision=0]{4370} queries \bl{without data tests}, which can be structurally divided into 5 groups (G1-G5 in \Cref{tab:experiments:GroupByStructure}). The following list shows an example query representative of each query group. 
{\small \begin{enumerate}[label=G\arabic*,leftmargin=2.5em]
    \item $\mathsf{q(x):=\node{Dataset}(x)\land has^*(x,y)\land\node{ReadingTask}(y)}$
    \item $\mathsf{q(x):=\node{Dataset}(x)\land has^*(x,y)\land\node{Lang\text{-}item}(y)}$
    \item $\mathsf{q(x):=\node{Dataset}(x)\land\mathsf{has^*(x,y_1)\land\node{Read}(y_1),}}\mathsf{has^*(x,y_2)\land\node{Lang\text{-}item}(y_2)}$
    \item $\mathsf{q(x):=has(x,y)\land\node{Read}(y)}$
    \item $\mathsf{q(x):=\node{ReadingTask}(x)}$
\end{enumerate}}
The queries in the groups G1-G3 request fMRI datasets, with either a specific cognitive task (G1), one specific HED tag (G2), or a combination of two HED tags (G3). 
Since the depth at which the task or event tags occur varies from an fMRI scan to another, the queries use the Kleene star to navigate to them. In the group {G4} and {G5} we query for individual HED tags and tasks.
\co{While there are early research prototypes that can parse and evaluate GQL queries \cite{GQLParser}, to the best of our knowledge, there are no publicly available, robust and scalable database systems, which support GQL at the time of our experimental evaluation. Having such systems}
would allow us to evaluate our queries under the walk-based semantics that coincides with the certain answer semantics, and Cypher only supports the so-called \emph{trail semantics} \change{\cite{Francis2018}. 
The walk-based semantics \cite{Angles2017} returns all nodes that match the RPE of a query, while the trail semantics does not visit the same edge twice.}
Through careful manual inspection, we generated queries for which both semantics coincide. 
Finding syntactic conditions for the two semantics to match is left for future work. \smallskip
{\newline \noindent
\textbf{Results.} In \Cref{tab:experiments:GroupByStructure} and \Cref{tab:experiments:GroupByQuerySize}, we report the results of the experiments grouped by the type of input query and the number of queries in the output union, respectively.}
In each table we provide the number of queries in that group, the average rewriting and evaluation time, as well as the average number of atoms in the rewritten query. Lastly, we state how often the evaluation timed out at \num[round-precision=0]{600}s. We averaged the times for rewriting and evaluation over 10 runs for each of the input queries. 
Constructing the \cdg, on which \Cref{algo:rewriteC2RPQ} depends, takes around two minutes.
In \Cref{tab:experiments:GroupByStructure} we can see that group G3, which has queries with a combination of two HED tags, takes rather long (on average more than \num[round-precision=0]{50} seconds). We attribute that to the number of atoms, which suggest that the output queries are larger compared to the queries in the other groups.
In \Cref{tab:experiments:GroupByQuerySize} we see that queries producing a smaller number of C2RPQs in the output union were evaluated faster (compare average evaluation time for the groups 1-10 to 30+). 
The runtime also increases with the number of atoms in the query, which in turn grows with the query size. 
The time it takes to rewrite the queries is on average below 6 seconds, even for the group with the largest queries. 
The evaluation time seems to be independent of the number of answers. 
Additionally, we ran experiments with a version of \Cref{algo:rewriteC2RPQ} that does not check for structural query subsumption in \cref{line:checkContainment}. However, for our use case this rewriting algorithm often produces a union with more than \num[round-precision=0]{2000} C2RPQs. As a consequence 
the evaluation often times out, so this is no longer practicable. 

\section{Conclusion and Future Work} \label{sec:conclusion}
We presented an algorithm for rewriting \restrictedQuery{s} into UC2RPQs over a lightweight ontology that extends DL-Lite with some of the expressive features of $\mathcal{ELH}$ while keeping the data complexity of reasoning in \NL.  
Our restricted input query language (\restrictedQuery{s}) is justified by the fact that we have proven the impossibility of using (U)C2RPQs as both the input and output language of query rewriting. Nested regular path queries seem a promising target language for rewriting ontology-mediated C2RPQs. 
\todo[author=\bf R3,color=green!20]{it should probably state, “Nested regular path
    queries seem a promising source language for rewriting”. \textbf{BL}: I think our sentence is correct, but maybe it is not clear what source and target query language means?}
One of our goals is 
to find
a practicable algorithm, and our 
prototype implementation, which rewrites the queries into Cypher, suggests that we may be on track. 
It shows promising results on 
a real-world dataset from cognitive neuroscience. 
For future work, we aim to support a richer ontology language, and to target 
the GQL standard to support full C2RPQs. 

\begin{credits}
\subsubsection{\discintname} The authors have no competing interests to declare that are relevant to the content of this article. 

\subsubsection{\ackname} This research was funded in whole or in part by the Austrian Science Fund (FWF) PIN8884924 and P30873. This work was also partially supported by the Wallenberg AI, Autonomous Systems and Software Program (WASP) funded by the Knut and Alice Wallenberg Foundation. This work was partially supported by the State of Salzburg under grant number 20102-F2101143-FPR (DNI) and the Austrian Federal Ministry of Education, Science and Research (BMBWF) under grant number 2920 (Austrian NeuroCloud). The authors acknowledge the computational resources and services provided by Salzburg Collaborative Computing (SCC), funded by the Federal Ministry of Education, Science and Research (BMBWF) and the State of Salzburg.


\end{credits}
%
%
%

\bibliographystyle{splncs04}
\bibliography{omq-pg-db}

\newpage 

\appendix

\section{Additional Preliminaries}
\label{app:Prelmin}

We provide here some further materials, that are not critical to be explicitly given in \Cref{sec:preliminaries} since they are standard notions from the cited works given in that section.

\begin{table}
    \caption{Semantics of $\mathcal{ELHI}$}
    \label{tab:elhi}
    \centering
    \begin{tabular}{l|l|l}
        \hline
        Name 				& Syntax 		& Semantics \\
        \hline
            top concept 	& $\top$		& $\varDelta^\mathcal{I}$ \\
        concept name 		& $A$ 			& $A^\mathcal{I}\subseteq \varDelta^\mathcal{I}$ \\
        negation 			& $\lnot A$ 	& $\varDelta^\mathcal{I}\setminus A^\mathcal{I} $ \\
        role name 			& $r$			& $r^\mathcal{I}\subseteq \varDelta^\mathcal{I}\times\varDelta^\mathcal{I}$ \\
        inverse role		& $r^-$			& $\{(b,a) \mid (a,b) \in r^\I \}$ \\
        exist. restriction 	& $\exists r.C$ & $\{a \mid b\in C^\mathcal{I}:(a,b)\in r^\mathcal{I}\}$\\
        conjunction         & $C\sqcap D$   & $C^\I\cap D^\I$ \\
        concept inclusion 	& $C \ISA D$	& $C^\mathcal{I} \subseteq D^\mathcal{I}$ \\
        \hline
    \end{tabular}
\end{table}

\begin{definition}[Semantics of C2RPQs]
	Consider a C2RPQ $q(\vec{x})$ and 
    a property graph $\A=(N,E,\mathsf{label, prop})$.
	We define the evaluation of $q(\vec{x})$ in terms of a function $\eval{\cdot }{\A}$, given in \Cref{tab:eval}, which assigns to every query a set of mappings from $\mathit{vars}(q)$ to nodes in $N$.   
	A tuple $\vec{a}\subseteq N$ is an \emph{answer} of $q(\vec{x})$ in $\A$ iff there exists a mapping $\mu\in\eval{\varphi}{\A}$ such that $\mu(\vec{x})=\vec{a}$. 
\end{definition}

\begin{table}
    \caption{Walk-based evaluation function for C2RPQs 
    (\cite{Perez2006}) extended with data tests.}
\label{tab:eval}
    \centering
    \begin{tabular}{l l}
        \multicolumn{2}{l}{$\eval{k\! \odot\! v\ (x)}{\A}=\{\mu \mid v \text{ and }\textsf{prop}(\mu(x),k) \text{ in } \mathbf{D}, \co{\textsf{prop}(\mu(x),k)  \odot  v} \text{ holds}\}$} \\
        \multicolumn{2}{l}{$\eval{k\! \odot\! v\ (x,y)}{\A}=\{\mu \mid v \text{ and }\textsf{prop}((\mu(x),\mu(y)),k) \text{ in } \mathbf{D}, \co{ \textsf{prop}((\mu(x),\mu(y)),k) \odot v }  \text{ holds}\}$} \\
        $\eval{T\land T'(x)}{\A}=\eval{T(x)}{\A}\cap\eval{T'(x)}{\A}$ & 
        $\eval{T\lor T'(x)}{\A}=\eval{T(x)}{\A}\cup\eval{T'(x)}{\A}$ \\
        $\eval{T\land T'(x,y)}{\A}=\eval{T(x,y)}{\A}\cap\eval{T'(x,y)}{\A}$ & 
        $\eval{T\lor T'(x,y)}{\A}=\eval{T(x,y)}{\A}\cup\eval{T'(x,y)}{\A}$ \\
        $\eval{\lnot T(x)}{\A}=\{\mu\mid \mu\not\in \eval{T(x)}{\A}\}$ & $\eval{\lnot T(x,y)}{\A}=\{\mu\mid \mu\not\in \eval{T(x,y)}{\A}\}$ \\
        \multicolumn{2}{l}{$\eval{\node{A}(x,y)}{\A}=\{\mu \mid A\in\mathsf{label}(\mu(x)), \mu(x)=\mu(y)\}$ } \\
        $\eval{r(x,y)}{\A}=\{\mu\mid r(\mu(x),\mu(y))\in E\}$ & $\eval{r^-(x,y)}{\A}=\{\mu\mid r(\mu(y),\mu(x))\in E\} \quad$ \\
        \multicolumn{2}{l}{$\eval{(\pi\pi')(x,z)}{\A}=\{\mu \cup \mu'\mid \mu\in\eval{\pi(x,y)}{\A}, \mu'\in \eval{\pi'(y,z)}{\A},\mu(y)=\mu'(y)\}$}\\ 
        \multicolumn{2}{l}{$\eval{(\pi\union\pi')(x,y)}{\A}=\eval{\pi(x,y)}{\A} \cup \eval{\pi'(x,y)}{\A}$} \\
        \multicolumn{2}{l}{$\eval{\pi^*(x,y)}{\A}=\{\mu\mid \mu(x)\in N, \mu(x)=\mu(y)\}\cup\eval{\pi(x,y)}{\A}\cup \eval{\pi\pi(x,y)}{\A}\cup\dots$}
    \end{tabular}
\end{table}

\newpage
\section{Algorithm for rewriting NCQs}\label{app:Algorithm2}

We present the rewriting for NCQs in \Cref{algo:rewriteC2RPQ}. 

As input we assume an \restrictedQuery~ $q(\vec{x})$ and an $\ontoLang$ TBox $\T$. 
We iterate over all variable subsets $Y$ and axioms in $\T$ of the form \ref{NF4} and \ref{NF8}, and apply the clipping function exhaustively.  
Then
we loop over the queries inside $q' \in Q$ and for each atom $A_i$ of a union $A_1\union\dots A_i\dots\union A_n(x)$, we compute its witnessing sets and produce a new query by producing a conjunction  using $\mathtt{rewrConcept}$ as shown. For roles $r$ occurring inside $q'$, we similarly produce a new query, replacing each occurrence of $r$ with the output of $\mathtt{rewriteRole}$. 
Note here that we use the function $\mathtt{add}^{\subseteq}$ to avoid adding any rewritten queries to our output which are structurally contained by another element in the output. 


\begin{algorithm}
\small
    \caption{Function $\mathtt{rewrite\restrictedQuery}$ 
	}
    \label{algo:rewriteC2RPQ}
    \SetKwInOut{Input}{Input}
    \SetKwInOut{Output}{Output}

    \SetKwFunction{Fclipper}{$\mathtt{clipping}$}
    \SetKwFunction{FrewriteRole}{$\mathtt{rewriteRole}$}
    \SetKwFunction{FrewriteConcept}{$\mathtt{rewrConcept}$}
    \SetKwFunction{Funion}{$\mathtt{add}^\subseteq$}
    \SetKwFunction{FrewriteNCQ}{$\mathtt{rewriteNCQ}$}
    \SetKwProg{Fn}{function}{:}{}
	
    \Input{\restrictedQuery \xspace $q = (\varphi,\vec{x})$, $\ontoLang$ TBox $\T$}
    \Output{UC2RPQs $Q$}
    
    \Fn{\FrewriteNCQ{$q,\T$}}{ \label{line:clippingLoopStart}
        $Q:=\{q \}, Q':=\emptyset$ \\
        \While{$Q'\neq Q$}{
              $Q':=Q$ \\
              \ForEach{$q'\in Q'$}{
                 \ForEach{$A\ISA \exists r.B\in\T$ and $Y\subseteq \mathit{vars}(q'), Y\cap\vec{x}=\emptyset$\label{line:variableSet}}{
                     $Q := Q \  \cup \  $\Fclipper{$q',A\ISA \exists r.B\in\T,Y$}
                 }        
             }
        } \label{line:clippingLoopFinish}

        \ForEach{$q'\in Q'$}{
            \ForEach{$A_1\union\dots A_i\dots\union A_k(x)$ occurring in $q'$}{
           \ForEach{$\{B_1,\dots,B_n\}\in$ \FWitness{$A_i,G_\T$}\label{line:loopWitnesses}}{
                    $q'_{\mathit{rw}}(x):=${\footnotesize {\scriptsize\FrewriteConcept{$B_1,G_\T$}}$(x,y_1) \land \dots \land $ {\scriptsize\FrewriteConcept{$B_n,G_\T$}}$(x,y_n)$ }\label{line:rewriteConcept}\\
                    $Q:=$\Funion{$Q,q'_{\mathit{rw}}$} \label{line:checkContainment} \\
                }
             }
            \ForEach{roles $r$ occurring in $q'$}{
                $Q:=Q\cup q'[r \backslash $\FrewriteRole{$r,\T$}$]$ \label{line:rewriteRole}
            }
         }
         \Return $Q$
    } 

    \Fn{\Funion{$Q,q$}}{
        \ForEach{$q'\in Q$}{
            \If{$q\subseteq q'$}{ \label{line:containsQuery}
                \Return $Q$
            }
            \If{$q'\subseteq q$}{
                $Q:= Q\backslash q'$ \label{line:dropQuery}
            }
        }  
        \Return $Q\cup q$
    } 
\end{algorithm}

\newpage

\section{Proofs \Cref{sec:tboxreasoning} Standard Reasoning in $\ontoLang$}
\label{app:ProofsSection StandardReasoning}




\begin{appendixLemma}{\ref{lemma:SoundnessPropagatingPaths}}
[Soundness Propagating Paths]
    Let $\T$ be an $\ontoLang$ TBox and let $G_\T$ be its \cdg. If there is a proper witnessing set $W_A=\{A_1,\dots,A_n\}$ for $A$ and a set of concepts $\{C_1, \ldots, C_n\}$ (of the form $\some s_1\dots \some s_k.B$ with $k\ge 0$) such that, for each $1\le i \le n$, there is a path in $G_\T$ from $A_i$ propagating $C_i$, then $\T \models C_1\AND\dots\AND C_n \ISA A$. 
\end{appendixLemma}

\begin{proof}
    Consider an $\ontoLang$ TBox $\T$ in normal form and the \cdg $G_\T$. 
    Assume that there is a proper witnessing set $\{A_1,\dots,A_n\}$ for $A$ and a set of concepts $\{C_1, \ldots, C_n\}$ (of the form $\some s_1\dots \some s_k.B$ with $k\ge 0$) such that, for each $1\le i \le n$, there is a path in $G_\T$ from $A_i$ propagating $C_i$. 
    We first show for each of these paths that $\T\models C_i\ISA A_i$ holds and consider the witness set and the conjunction in a second step. 
    
    By \Cref{def:CDG_nodes_and_edges} the path propagating $C_i$ has the role label $(r_1\ldots r_k)$ and $\transclosure{r_j}{s_j}$ for $1\le j\le k$.
    We show the claim by induction on $k$. 
    First, consider the base case $k=0$, where the path is a single node $B\in\conceptnames$. Since $\T\models B\ISA B$ is always valid the claim holds for this case.
    
    For the induction step $k+1$, consider the path $A_i\rho_1\dots\rho_l B_l \rho_{l+1} B$ in $G_\T$ and suppose that $A_i\rho_1\dots\rho_l B_l$ propagates $\exists s_1\dots\exists s_k.B_l$. 
    From the induction hypothesis it holds that $\T\models \exists s_1\dots\exists s_k.B_l\ISA A_i$. 
    Then, we distinguish two cases, either $\rho_{l+1}$ is (1) a $r_{k+1}$-propagating edge with $\transclosure{r_{k+1}}{s_{k+1}}$, or (2) an $\varepsilon$-propagating edge. 
    
    In case (1) the $s_{k+1}$-propagating edge from $B_l$ to $B$ was either added by \cref{def:conceptdependencygraph:item:NF3} or \cref{def:conceptdependencygraph:item:roleAndInverse}. For \cref{def:conceptdependencygraph:item:NF3} $\T\models \exists s_{k+1}.B\ISA B_l$ follows directly from the CI in $\T$. In case of \cref{def:conceptdependencygraph:item:roleAndInverse}, $B$ is the $\top$ concept and there is a propagating path with role label $s_{k+1}$ from $B_l$ to some node $D$ and $\exists r_{k+1}^-.\top\ISA D$ with $\transclosure{s_{k+1}}{r_{k+1}}$. By induction hypothesis we assume that for this propagating path $\T\models \exists r_{k+1}.D\ISA B_l$ holds, thus, $\T\models\exists r_{k+1}.\top\ISA B_l$. Now, for both items, from $\T\models \exists r_1\dots\exists r_k.B_l\ISA A_i$, $\T\models \exists r_{k+1}.B\ISA B_l$ and $\transclosure{r_i}{s_i}$ with $1\le i\le k+1$ it follows that $\T\models \exists r_1\dots\exists r_{k+1}.B\ISA A_i$.  
    
    In case (2) the $\varepsilon$-propagating edge from $B_l$ to $B$ must be added by either \cref{def:conceptdependencygraph:item:top}, \ref{def:conceptdependencygraph:item:NF1}, \ref{def:conceptdependencygraph:item:closeconjunction}, \ref{def:conceptdependencygraph:item:NF4}, \ref{def:conceptdependencygraph:item:NF7_inv}, \ref{def:conceptdependencygraph:item:NF7}, or \ref{def:conceptdependencygraph:item:close}.
    \begin{enumerate}[topsep=0pt]
        \item[\ref{def:conceptdependencygraph:item:top}] $B_l$ is the $\top$ concept and since $\T\models B\ISA \top$ holds for any concept, $\T\models \exists r_1\dots\exists r_k.B\ISA A_i$ follows. 
        \item[\ref{def:conceptdependencygraph:item:NF1}] $B$ is of the form $D_1\sqcap\cdots\sqcap D_m$ and the claim holds since the CI $D_1\sqcap\dots\sqcap D_m\ISA B_l$ is in $\T$. 
        \item[\ref{def:conceptdependencygraph:item:closeconjunction}] let's assume that there is a node $D_1\sqcap\cdots \sqcap D_m$ and there are $\varepsilon$-propagating edges from $D\in\{D_1,\dots,D_m\}$ to $B_l$ in $G_\T$. Then, we want to show that $\T\models B_l\ISA D_1\sqcap \dots\sqcap D_m$ holds. From the propagating edges we suppose that $\T\models B_l\ISA D$ for each $D\in\{D_1,\dots,D_m\}$ by induction hypothesis. 
        It follows that the claim holds, since $B_l\ISA D_1\sqcap\dots\sqcap D_m$ is just another notation for it.
    \end{enumerate}
    In case that the edge was added through \cref{def:conceptdependencygraph:item:NF4}, \cref{def:conceptdependencygraph:item:NF7_inv}, or \cref{def:conceptdependencygraph:item:NF7}, $B$ is of the form $W_{\exists r.B_l}$.
    Since $W_{\exists r.B_l}$ is not a concept in $\T$ we treat it as a placeholder that serves as an anonymous witness that is in the extension of $B_l$. 
    Put another way, it substitutes a conjunction of all concept names that have an outgoing $\varepsilon$-propagating edge to $W$. 
    \begin{enumerate}[topsep=0pt]
        \item[\ref{def:conceptdependencygraph:item:NF4}] $B$ is of the form $W_{\exists r.B_l}$ and there has to be some axiom of the form $C\ISA\exists r.B_l$ in $\T$ and $W_{\exists r.B_l}$ is simply the placeholder for $B_l$.
        \item[\ref{def:conceptdependencygraph:item:NF7_inv}] we know that $\T\models \exists r^-.\top\ISA B_l$ and again some axiom of the form $\T\models C\ISA\exists r.D$ from the $r$-existential edge. From these two axioms and the placeholder $W_{\exists r.D}$ we can safely replace $W_{\exists r.D}$ by $W_{\exists r.D}\sqcap D$.
        \item[\ref{def:conceptdependencygraph:item:NF7}] works analogous to the previous item, just for the inverse roles.
        \item[\ref{def:conceptdependencygraph:item:close}] we add $\varepsilon(B_l,B)$ to $G_\T$ if there is a node $W$ and roles $r,s\in\allroles$ with $\transclosure{s}{r}$ such that $\exists s(B,W)$ in $G_\T$ and there is a propagating path from $B_l$ to $W$ with role label $r$. 
        Observe that $W$ is a placeholder for a conjunction and we assume that $\T \models B\ISA \exists s.W$ holds. 
        Since there is a propagating path from $B_l$ to $W$ with role label $r$ we suppose by that $\T\models \exists r.W\ISA B_l$ holds. 
    \end{enumerate}
    To sum it up, since $\T\models B\ISA \exists s.W$, $\T\models\exists r.W\ISA B_l$ and $\transclosure{s}{r}$, we can infer that $\T\models B\ISA B_l$. 
    Thus, it holds that if there is a path from $A_i$ propagating $\exists s_1\dots\exists s_{k}.B$, then $\exists s_1\dots\exists s_{k}.B\ISA A_i$ holds. 

    It remains to show that $\T\models C_1\AND\dots\AND C_n\ISA A$ holds for a set $W_A=\{A_1,\dots,A_n\}$ that witnesses $A$ such that from each $A_i\in W_A$ there is a path that propagates $C_i$. Above we already showed that $\T\models C_i\ISA A_i$ holds. 
    In the following we show $\T\models A_1\sqcap\dots\sqcap A_n\ISA A$ by induction on the witnessing sets $\mathbf{W}_A=\{W_1,\dots,W_m\}$. 
    For the base case $m=0$, $A\in W_A$ (\Cref{def:witnessingSet} item (a)), for which the claim trivially holds, since $\T\models A\sqcap \dots\ISA A$. 
    For the induction step $m+1$, consider the set $W_{m+1}=\{A_1,\dots,A_n\}$ witnessing $A$ and by induction hypothesis we can assume that for the witnessing set $W_m=\{D_1,\dots,D_j\}$ of $A$, $\T\models D_1\sqcap\dots\sqcap D_j\ISA A$ holds. 
    Then, by \Cref{def:witnessingSet} it must be that for each $A_i\in W_{m+1}$ either (a) $A_i\in W_m$, or (b) there is a propagating path from $D\in W_m$ in $G_\T$ to a conjunction containing $A_i$. Note that by \Cref{def:normalform} this path can only contain $\varepsilon$ edges. 
    By putting these facts together it holds for the set $W_m$ that $\T\models A_1\AND \dots \AND A_n\ISA A$. 

    To sum it up, for the set $W_A=\{A_1,\dots,A_n\}$ it follows that $\T\models A_1\AND \dots \AND A_n\ISA A$ and for each $A_i\in W_A$ it holds that $\T\models C_i\ISA A_i$. Thus we can conclude that the claim $\T\models C_1\AND \dots \AND C_n\ISA A$ holds. 
\end{proof}
\newpage

\begin{definition}[Model Witnessing Propagating Paths]\label{def:modelConstruction}
    For a given ABox $\A$ and a TBox $\T$. Let $G_\T$ be the \cdg of $\T$. We define $\I_0$ as follows: 
    \begin{align*}
            \Delta^{\I_0} =& \{  a \mid A(a) \text{ or } r(a,b) \text{ or } r(b,a) \text{ is in } \A  \} \\
            & \cup \{ [\exists r . D ]   \text{ in $G_\T$ such that } \exists r (A, \exists r.D) \in E_\T \} \\
            A^{\I_0}=& \{ [W_{\exists r.D}] \mid \text{there is a path from } A \text{ in } G_\T \text{ propagating } W_{\exists r . D} \} \\ 
            & \cup \{ a \mid B_1(a),\dots, B_n(a) \in \A \text{ and there is } \\ 
            & \quad \ \text{a set } \{B_1,\dots,B_n\} \text{ witnessing } A \}\\
            r^{\I_0} = &\{  (a,b) \mid s(a,b) \in A \text{ for some } s \sqsubseteq^*_\T r \} \\ 
            & \cup \{  (a, [W_{\exists s. B} ]) \mid a \in A^\I \text{ for some } A\in N_\T, \\
            & \quad \ \ \text{some } \exists s (A, W_{\exists s.B}) \in G_\T \text{ and } s \sqsubseteq^*_\T r \} \\
            & \cup \{  ([W_{\exists s^-. \top}],a) \mid a \in A^\I \text{ for some } A\in N_\T, \\
            & \quad \ \ \text{some } \exists s^- (A, W_{\exists s^-.\top}) \in G_\T \text{ and } s \sqsubseteq^*_\T r \}
    \end{align*}
    Then, we obtain a model $\I$ by chasing $\I_0$ under the following two rules until we reach a fixpoint. 
    \begin{enumerate}
        \item\label{def:modelConstruction:conjunction} If $a \in (A_1\AND \cdots \AND A_n)^{\I_i}$ and $\{A_1,\dots,A_n\}$ witnesses $A$,
        then add $a$ to $A^{\I_{i+1}}$.
        \item\label{def:modelConstruction:propagatingpath} If $a \in (\exists r_1\dots\exists r_k . B)^{\I_i}$ and there is a path from $A$ that propagates \\
        $\exists r_1\dots\exists r_k . B$, then add $a$ to $A^{\I_{i+1}}$.
        \item\label{def:modelConstruction:existentialedge} If $a \in A^{\I_i}$ and there is $\exists r(A,[W_{\exists r . B}])\in G_\T$,  
        then add $(a, [W_{\exists r . B}]) \in s^ {\I_{i+1}} $ for all $ r \sqsubseteq^*_\T s$.
    \end{enumerate} 
\end{definition}

\begin{lemma}[Model Witnessing Propagating Paths]\label{lemma:validModel}
    Given an $\ontoLang$ TBox $\T$ in normal form and an arbitrary ABox $\A$. Then, the interpretation $\I$ constructed according to \Cref{def:modelConstruction} is a model of $\T$ and $\A$. 
\end{lemma}

\begin{proof}
    Consider an $\ontoLang$ TBox $\T$ in normal form, an ABox $\A$ and $\I$ the interpretation constructed according to \Cref{def:modelConstruction}. 
    From the definition for $\I_0$ in \Cref{def:modelConstruction} it holds that $\A$ is a subgraph of $PG(\I)$, thus $\I$ is a model of $\A$ according to \Cref{def:PropInterpret}.
    For $\I$ to be a model of $\T$, we show in the following that an arbitrary axiom in $\T$ is satisfied. 
    \begin{enumerate}[label=\textit{(NF\arabic*)}, leftmargin=3.2em]
        \item $C_1\sqcap\dots\sqcap C_n \ISA D$. Let $c\in C_i^\I$ for $1\le i\le n$, it follows by item \ref{def:conceptdependencygraph:item:NF1} that $\varepsilon(D,C_1\sqcap\dots\sqcap C_n)\in G_\T$ and the set $\{C_1,\dots,C_n\}$ witnessing $D$. From the model construction $\I$ in \Cref{def:modelConstruction} \cref{def:modelConstruction:conjunction} it holds that $c\in D^\I$.
        \item $\exists r.C\ISA D$. Let $c\in (\exists r.C)^\I$, from the construction of the \cdg there is $r(D,C)$ in $G_\T$ and by construction of $\I$ \Cref{def:modelConstruction} \cref{def:modelConstruction:propagatingpath} it follows that $c\in D^\I$.
        \item $C \ISA \exists r.D$. Let $c\in C^\I$, it follows from the construction of the \cdg that $\exists r(C,W_{\exists r.D})$ and $\varepsilon(D,W_{\exists r.D})$ in $G_\T$. From $\exists r(C,W_{\exists r.D})$ and $c\in C^\I$  it holds by \Cref{def:modelConstruction} \cref{def:modelConstruction:existentialedge} that $(c,[W_{\exists r.D}])$ in $r^\I$. Further, from $\varepsilon(D,W_{\exists r.D})$ and the construction of $\I_0$ it follows that $[W_{\exists r.D}]\in D^\I$ and we conclude that $c\in (\exists r.D)^\I$.
        \item $s\ISA r$. Let $(c,d)\in s^\I$, since we always consider the transitive closure, when adding edges to $\I$ in \Cref{def:modelConstruction}, it holds that $(c,d)\in r^\I$. 
        \item $\exists r^-.\top\ISA D$. Let $c\in (\exists r^-.\top)^\I$, from the construction of the \cdg \cref{def:conceptdependencygraph:item:NF3} it holds that there is a path from $D$ propagating $\exists r^-.\top$. Thus, by \cref{def:modelConstruction:propagatingpath} it follows that $c\in D^\I$ and the claim holds.
        \item $C\ISA \exists r^-.\top$. Let $c\in C^\I$, by the construction of the \cdg \cref{def:conceptdependencygraph:item:NF4} there is an $\exists r^-(C,W_{\exists r^-.\top})$ and $\varepsilon(\top,W_{\exists r^-.\top})$. By construction of $\I$ \cref{def:modelConstruction:existentialedge} $c\in (\exists r^-.\top)^\I$ and the claim holds.
    \end{enumerate}
\end{proof}


\begin{appendixLemma}{\ref{lemma:completenessPropagatingPaths}}[Completeness Propagating Paths]
    Let $\T$ be a $\ontoLang$ TBox and let $G_\T$ be its \cdg. 
    If $\T \models C_1\AND\dots\AND C_n \ISA A$ where $C_1,\dots, C_n$ are concepts of the form $\some s_1\dots \some s_k.B$ with $k\ge 0$ and $A,B\in\conceptnames$, then there is a set $W_A = \{A_1, \cdots, A_n\}$ that witnesses $A$ and such that for each $1 \leq i \leq n$ there is a path in $G_\T$ from $A_i$ that propagates $C_i$.
\end{appendixLemma}

\begin{proof} We show the claim by contraposition, for that we construct a model $\I$ according to \Cref{def:modelConstruction} with a given TBox $\T$. For the ABox consider $\T\models C_1\AND \cdots C_j\cdots \AND C_n\ISA A$, then we compose $\A=\A_1\cup\cdots \A_j\cdots\cup \A_n$ such that $\A_j=\{r_i(n_{i-1},n_i)\mid C_j=\exists r_1  \cdots \exists r_k.B_j \text{ and } 1\le i \le k\}\cup\{B_j(n_k)\mid C_j=\exists r_1  \cdots \exists r_k.B_j\}$.
From \Cref{lemma:validModel} it holds that $\I$ is a model of $\T$. 
Since we construct $\I$ in such a way that it makes true exactly the inclusions witnessed by the witness set with the propagating paths, it remains to show that $A(n_0)\not\in \I$ if there is no set $W_A=\{A_1,\dots,A_n\}$ that witnesses $A$ such that for each $A_j\in W_A$ there is a path from $A_j$ propagating $\exists r_1  \cdots \exists r_k.B_j$.
Since $n_0$ has to be in $(A_1\AND\dots\AND A_n)^{\I_i}$ to be added to $A^{\I_{i+1}}$ by \cref{def:modelConstruction:conjunction}, it is enough to show that $n_0\not\in A_j^\I$ if there is no path from $A_j$ propagating $\exists r_1  \cdots \exists r_k.B_j$.
In the following we show this claim by induction on $k$ and consider the ABox $\A_j$. 

Consider the base case $k=0$, then the ABox is of the form $\A_j=\{B_j(n_0)\}$ and we assume that there is no path from $A_j$ propagating $B_j$ in $G_\T$. From the fact that there is no such path from $A_j$ we can also infer that $\{B_j\}$ is not witnessing $A$. Thus, neither by construction of $\I_0$ nor by \cref{def:modelConstruction:conjunction} and \cref{def:modelConstruction:propagatingpath} we add $n_0$ to $B_j$.  

We continue with the induction step $k+1$, for that we consider the ABox $\A_j=\{r_i(n_{i-1},n_i)\mid 1\le i \le k+1\}\cup\{B_j(n_{k+1})\}$ and $B_j$ a concept name. Further, assume that there is no $\exists r_1\dots\exists r_{k+1}.B_j$ propagating path from $A_j$. We distinguish two cases, either from some node $D$ in $G_\T$ (1) there is an $\exists r_{k+1}.B_j$ propagating path, or (2) there is an $\exists r_1\dots\exists r_k.D$ propagating path but no $\exists r_{k+1}.B_j$ propagating path from $A_j$. For (1), it must hold that there is no $\exists r_1\dots\exists r_k.D$, since otherwise we would have an $\exists r_1\dots\exists r_{k+1}.B_j$ propagating path from $A_j$ against the assumption. By induction hypothesis it follows from the fact that there is no $\exists r_1\dots\exists r_k.D$ that $n_0\not\in A_j^\I$. 
For (2), based on the fact that there is no $\exists r_{k+1}.B_j$ propagating path from $D$ we want to show that $n_k\not\in D^\I$. 
From the fact that there is no $\exists r_{k+1}.B_j$ propagating path from $D$ we can infer that there is also no $\exists r_{k+1}.\top$ propagating path. Further, there must not be an $\exists r^-_k.\top$ propagating path from $D$, since this would mean we have an $\exists r_1\dots\exists r_k\exists r^-_k.\top$ and thus by \cref{def:conceptdependencygraph:item:roleAndInverse} an $\exists r_1\dots\exists r_k.\top$ propagating path. By definition of propagating paths it holds that there is an $\exists r_1\dots\exists r_k.B_j$ propagating path from $A_j$, which is against the assumption. 
To conclude, there is no propagating path in $G_\T$ that adds $n_k$ to $D^\I$ in the construction of $\I$ by \Cref{def:modelConstruction}. Therefore, even though there is an $\exists r_1\dots\exists r_k.D$ propagating path from $A$ we do not add $n_0$ to $A_j^\I$, since $n_k\not\in D^\I$.
\end{proof}

The following lemma considers a given ABox and almost immediately follows from \Cref{lemma:SoundnessPropagatingPaths} and \Cref{lemma:completenessPropagatingPaths}.

\begin{appendixLemma}{\ref{lemma:relation}}
    Let $\T$ be an $\ontoLang$ TBox \co{ and let $G_\T$ be its \cdg.} 
    For each ABox $\A$,  $A \in \conceptnames$ and individual $n_0$, we have 
    $\T,\A \models A(n_0)$ iff there is a proper witnessing set $W_A=\{A_1,\dots,A_n\}$ for $A$ such that for each $A_j\in W_A$ there exists 
    \textit{(i)} a (possibly empty) path $p_j$ \co{ in $G_\T$ }
    with role label $(r^j_1\cdots r^j_k)$ that propagates a concept $\some s_1 \dots\some s_k.B_j$; and 
    \textit{(ii)} a sequence of individuals $n^j_{1},\dots,n^j_{k}$ such that $B_j(n^j_{k}) \in \A$
    and  $r^j_i(n^j_{i},n^j_{(i+1)}) \in \A$
    for each $1 \leq i < k$.
\end{appendixLemma}

\section{Proofs Section Rewriting Navigational Queries}
\label{app:ProofsSectionRewriteNQs}

\begin{appendixTheorem}{\ref{thm:c2rpq_rewriting}}[Rewriting C2RPQs into UC2RPQs with $\exists r.\top$]
	There exist C2RPQs that cannot be rewritten into UC2RPQs w.r.t. TBoxes containing concepts of the form $\exists r.\top$ on both sides of CIs. 
    This holds already for C2RPQs with only one atom.
\end{appendixTheorem}

\begin{proof}
    Consider the C2RPQ $q(x,y)=(r p p^-)^+(x,y)$ (where $(r p p^-)^+$ stands for $(r p p^-)(r p p^-)^*$) and a 
    the TBox $\T=\{\exists s.\top \ISA \exists p.\top\}$. 
	A valid rewriting of $q$ w.r.t. $\T$ into a nested two-way regular path query would be $Q_r(x,y)=(r (p p^- \union \node{s}))^+(x,y)$, where we denote nested queries by "$\node{}$". In the following we show that it is not possible to express the same with the language of UC2RPQ, i.e.,
    there is no UC2RPQ $Q_r(x,y)$ such that, for every ABox $\A$, the answers to $Q_r$  over $\A$ coincide with the answers to $q$ over $\A$ and $\T$.
 
	For the sake of contradiction we assume that there is a UC2RPQ rewriting $Q_r(x,y)$ with $N$ the maximal number of conjuncts in a disjunct of $Q_r(x,y)$. 
	Then, consider the following ABox $\A$ (without the grayed-out nodes and edges)
	\begin{center}
		\includegraphics[width=.6\linewidth]{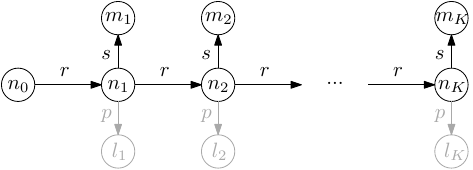}
	\end{center}
	where $K > N+1$. 
	Since we assume that there is a UC2RPQ $Q_r(x,y)$ and given that $(n_0,n_K)$ is a certain answer to $(\T,\A,q(x,y))$, there exists a disjunct $q_r(x,y)$ such that $(n_0,n_K)\in\eval{q_r(x,y)}{\A}$. 
	Then, we assume without loss of generality that the C2RPQ $q_r(x,y)$ is of the form $r_1(z_1,z_1')\land \dots \land r_N(z_N,z_N')$ for which there exists a mapping $\mu$ from $\{z_1,z_1',\dots,z_N,z_N'\}$ to $\{n_0,n_1,m_1,\dots,n_K,m_K\}$ such that $\mu(x)=n_0$, $\mu(y)=n_K$ and for $1\le i \le N$ there is a path $\pi_i$ between $\mu(z_i)$ and $\mu(z_i')$, such that the sequence of edge labels $L_i$ associated to $\pi_i$ satisfies the regular path expression $r_i$. 
	
	Consider the following steps to construct an ABox $\A_r$ from the mapping $\mu$ and $q_r$. 
    First, add all $\mu(z_i)$ and $\mu(z'_i)$ with $1\le i \le N$ as nodes to the ABox. 
    In a second step we extend the graph at $a_0=\mu(z_i)$ by fresh nodes and edges matching the regular path expression $r_i$. 
    For that consider the sequence of edges (with inverses) $L_i=e_0e_1\dots e_k$ from above, satisfying the regular path expression $r_i$. 
    Then, we add the nodes $a_1,\dots,a_k$ with $a_{k+1}=\mu(z'_i)$ and the following edges to the ABox, where $1\le j \le$ $k+1$
    \begin{itemize}
        \item if $e_j=r$ or $e_j=s$, then add $e_j(a_{j-1},a_j)$
        \item if $e_j=r^-$ or $e_j=s^-$, then add $e_j(a_j,a_{j-1})$.
    \end{itemize}
    By the construction it follows that $(n_0,n_K)\in\eval{q_r(x,y)}{\A_r}$ for $\mu(x)=n_0$, $\mu(y)=n_K$. 
    It remains to show that $(n_0,n_K)$ is not an answer to $q(x,y)$ over $(\T,\A_r)$, i.e. an answer over $\A_{\T,\A_r}$ the corresponding property graph to the canonical model of $(\T,\A_r)$.
    
    First assume that $n_0$ and $n_K$ are in the same connected component, because if $n_0$ and $n_K$ are not connected in $\A_{\T,\A_r}$, then obviously $(n_0,n_K)\notin\eval{q(x,y)}{\A_{\T,\A_r}}$. 
    Next, observe that $\A_r$ contains only edges from the construction step above that converted a sequence of edge labels into a \emph{semi-path} (a path disregarding the direction), that is \emph{linear}, i.e., the sum of out and in-degree for each node on the path is at most 2. This and the fact that $K>N+1$ with $N$ the number of conjuncts in $q_r$ implies that there are intermediate nodes that have at most 2 as sum of in- and out-degrees in $\A_r$. 
    Further, the sequences of edge labels $L_i$ consists of either $r$, $s$ or their inverses, but not label $p$. Thus, for the C2RPQ $q(x,y)=(r p p^-)^+(x,y)$ to have a mapping in $\A_{\T,\A_r}$, an edge with label $p$ must be added in the construction of the canonical model based on an edge with label $s$ in $\A_r$. But then every node on a path that has a mapping for $(r p p^-)^+(x,y)$ (besides start and end node) must have in-degree at least 1 (in-going edge with label $r$) and out-degree at least 2 (out-going edges with labels $r$ and $s$). Therefore, $(n_0,n_K)$ cannot be in $\eval{q(x,y)}{\A_{\T,\A_r}}$, which contradicts the initial assumption that there exists a UC2RPQ rewriting for $q(x,y)$ over the given TBox.
\end{proof}

\begin{appendixLemma}{\ref{lemma:queryContainment}}[Query Containment] 
    Let $\T$ be a TBox, $\A$ be an ABox, and $q_1(\vec{x})$, $q_2(\vec{x})$ be two NCQs, such that $q_1\subseteq_\T q_2$. Then, $\vec{a}$ is a certain answer to $(\T,\A,\break q_2(\vec{x}))$ if $\vec{a}$ is a certain answer to $(\T,\A, q_1(\vec{x}))$. 
\end{appendixLemma}

\begin{proof}
    Consider a TBox $\T$, an ABox $\A$ and two NCQs $q_1(\vec{x}),q_2(\vec{x})$, such that $q_1\subseteq_\T q_2$. Assume that $\vec{a}$ is a certain answer to $(\T,\A,q_1(\vec{x}))$. According to \Cref{def:certainAnswer}, this means that $\vec{a}$ is an answer to $q_1(\vec{x})$ in $PG(\I)$ for every model $\I$ of $\A$ and $\T$. 
    Now, for a proof by contradiction, suppose that there is no certain answer in $(\T,\A,q_2(\vec{x}))$. This means that there is at least one atom $\beta_1\union\dots\union\beta_j\union\dots\union\beta_m(x,y)$ in $q_2$ for which there is no mapping in $PG(\I)$ with $\I$ being any model of $\A$ and $\T$. 
    However, since 
    by \Cref{def:queryContainment} it holds that there exists an atom $\alpha_1\union\dots\union\alpha_i\union\dots\union\alpha_n(x,y)$ in $q_1$ such that for each $\alpha_i$ there is a $\beta_j$ and $\T\models\alpha_i\ISA\beta_j$. 
    By assumption there exists a mapping $\mu\in\eval{\alpha_1\union\dots\union\alpha_i\union\dots\union\alpha_n(x,y)}{PG(\I)}$ for every model $\I$ of $\A$ and $\T$. Since there is a $\beta_j$ for each $\alpha_i$ such that $\T\models\alpha_i\ISA\beta_j$ it holds that $\mu$ is also a valid mapping in $\eval{\beta_1\union\dots\union\beta_j\union\dots\union\beta_m(x,y)}{PG(\I)}$, which is a contradiction and we can conclude that the claim holds. 
\end{proof}

In the proofs for \Cref{lemma:clipperonestep:soundness} and  \Cref{lemma:clipperonestep:completeness} we make use of an universal model that we define in \Cref{def:canonicalmodel}.

\begin{definition}[Canonical Model]\label{def:canonicalmodel}
    Consider an $\ontoLang$ TBox $\T$ and an ABox $\A$. 
    Let $\I_0$ be the initial interpretation with $\A$ the corresponding property graph. 
    Then, we construct the \emph{canonical model} $\I_{\T,\A}$ by exhaustively applying the following rules to $\I_0$.
    \begin{enumerate}[label=(Ch\arabic*),leftmargin = 3em]
        \item If $A_1\sqcap\dots A_n \ISA B\in\T$, $v\in A_j^{\I_i}$ for $1\le j\le n$ and $v\notin B^{\I_i}$, then $\I_{i+1}$ is $\I_i$ with $B^{\I_{i+1}}=B^{\I_{i}}\cup\{v\}$. 
        \item If $\exists r.C \ISA B\in\T$, $v\in (\exists r.C)^{\I_i}$ and $v\notin B^{\I_i}$, then $\I_{i+1}$ is $\I_i$ with $B^{\I_{i+1}}=B^{\I_{i}}\cup\{v\}$.
        \item If $A \ISA \exists r.C\in\T$, $v\in A^{\I_i}$ and $v\notin (\exists r.C)^{\I_i}$, then $\I_{i+1}$ is $\I_i$ with a fresh $w$ in $\varDelta^{\I_{i+1}}$ and $C^{\I_{i+1}}=C^{\I_{i}}\cup\{w\}$. In case $r\in \rolenames$ then $r^{\I_{i+1}}=r^{\I_{i}} \cup\{(v,w)\} $, otherwise $r^{\I_{i+1}}=r^{\I_{i}} \cup \{(w,v)\}$.
        \item If $r\ISA s\in\T$, $(v,w)\in r^{\I_i}$ and $(v,w)\notin s^{\I_i}$, then $\I_{i+1}$ is $\I_i$ with $s^{\I_{i+1}}=s^{\I_{i}}\cup \{(v,w)\}$ in case $r\in \rolenames$, otherwise $s^{\I_{i+1}}=s^{\I_{i}}\cup \{(w,v)\}$.
    \end{enumerate}
\end{definition}

Like other lightweight DLs, in $\ontoLang$ we can rely on the existence of a \emph{universal model} 
$\I_{\T,\A}$ 
 that can be used for answering all queries, and 
 which can be constructed with a usual \emph{chase} procedure. 
The following lemma is standard. 

\begin{lemma}\label{lemma:answercanonical}
	For each $\ontoLang$ TBox $\T$, 
    ABox $\A$, 
    and C2RPQ $Q(\vec{x})$,  
 $\vec{a}\subseteq N$ is a certain answer to $(\T,\A,Q)$ iff $\vec{a}$ is an answer to $Q(\vec{x})$ in $PG(\I_{\T,\A})$. 
\end{lemma}

\begin{lemma}[Soundness Clipping One Step]\label{lemma:clipperonestep:soundness}
    Let $\T$ be an $\ontoLang$ TBox and $\A$ an ABox consistent with $\T$. For two \restrictedQuery{s} $q_1(\vec{x})$ and $q_2(\vec{x})$ with $q_2(\vec{x})=\mathtt{clipping}(q_1(\vec{x}), A\ISA\exists r.B,y)$ such that $A\ISA\exists r.B\in\T$ and $y\not\in\vec{x}$, it holds that $\vec{a}$ is an answer to $q_2(\vec{x})$ over $(\T,\A)$ implies $\vec{a}$ is an answer to $q_1(\vec{x})$ over $(\T,\A)$.
\end{lemma}

\begin{proof}
    Assume that we obtain $q_2(\vec{x})=\mathtt{clipping}(q_1(\vec{x}), A\ISA\exists r.B,Y)$ where $\mu_2$ is a mapping for $q_2$ in $\A_{\I_{\T,\A}}$, the corresponding property graph of the canonical model $\I_{\T,\A}$ and $\vec{a}=\mu_2(\vec{x})$, i.e., $\mu_2\in\eval{q_2}{\A_{\I_{\T,\A}}}$. Moreover, let $a=\mu_2(y)$. Then, since by step (C6) we add $A(y)$ to $q_2$ we know that $a\in(\exists r.B)^{\I_{\T,\A}}$. From this follows that there has to be a $b\in \varDelta^{\I_{\T,\A}}$ such that $(a,b)\in r^{\I_{\T,\A}}$ and $b\in B^{\I_{\T,\A}}$. We define the mapping $\mu_1$ for the variables of $q_1$ as given below:
    \begin{enumerate}[label=(\alph*)]
        \item \label{proof:clipping:soundness:mapping1} $\mu_1(x)=a$ for all variables $x\in X$
        \item \label{proof:clipping:soundness:mapping2} $\mu_1(y)=b$ for all variables $y\in Y$
        \item \label{proof:clipping:soundness:mapping3} $\mu_1(z)=\mu_2(z)$ for the remaining variables z.
    \end{enumerate}
    In the following we show that $\mu_1$ is a mapping for $q_1$ in $\A_{\I_{\T,\A}}$. For that consider an atom $\alpha$ in $q_1$, then there are two cases: 
    \begin{enumerate}[label=(\alph*)]
        \item $y$ occurs in $\alpha$, then there are three further cases: 
            \begin{enumerate}[label=(\arabic*)]
                \item $\alpha$ contains a star and if there is a variable $z\neq y$, it is unbound. It follows that $z\not\in X$ and also doesn't occur in $q_2$, observe that $\mu_1(z)=b$, and $\mu_1(y)=b$ by construction of $\mu_1$. It is a valid mapping since $\alpha$ contains a star and $y$ can trivially be mapped to any node. 
                \item $\alpha$ contains a concept name $C$ with $\T\vDash B\ISA C$ or $\T\vDash\exists r^-\ISA C$ and if has a variable $z\neq y$ it is unbound. Since {$z\not\in X$} and also doesn't occur in $q_2$, {by construction of $\mu_1$ we have} $\mu_1(z)=b$ and $\mu_1(y)=b$. From above we know that $b\in B^{\I_{\T,\A}}$, from which follows that $b\in C^{\I_{\T,\A}}$. Hence, $\mu_1$ is a valid mapping for the atom $\alpha$. 
                \item $\alpha$ is either of the form $\pi(x,y)$ where $\pi$ contains $s$ with $\transclosure{r}{s}$ or $\pi(y,x)$ with $s^-$, and $x\neq y$. In both cases $\alpha\in \mathbf{C}_y$, then $\alpha$ can either be in $\mathbf{C}_y^1$ or $\mathbf{C}_y^*$. 
                Let's first consider the case where $\alpha\in \mathbf{C}_y^1$, which means that $s$ does not occur in the scope of a star. Further, $x\in X$ and thus $\mu_1(x)=a$ and $\mu_1(y)=b$. From above we know that $(a,b)\in r^{\I_{\T,\A}}$ and since $\transclosure{r}{s}$, we conclude that $\mu_1$ is a valid mapping for $q_1(\vec{x})$ in both cases, i.e. $\pi$ contains $s$ or $s^-$. 
                Next, we show the claim for the case where $\alpha\in \mathbf{C}_y^*$, which means that a role $s$ or $s^-$ with $\transclosure{r}{s}$ occurs in the scope of a star, i.e., $s^*$ or $(s^-)^*$. In (C5) we replace $\pi(x,y)$ (or $\pi(y,x)$) by $\pi\mid^\T r(x,y)$ (resp. $\pi\mid^\T r(y,x)$) to obtain $q_2$. According to the definition for $\pi\mid^\T r(x,y)$, $s^*$ (resp. $(s^-)^*$) remains in $\pi\mid^\T r(x,y)$. Hence, it holds that there is a path from $\mu_2(x)$ to $\mu_2(y)$ matching $\pi\mid^\T r(x,y)$ and therefore also $\pi(x,y)$. 
                Moreover, by assumption we have $\mu_2(y)=a$ and $(a,b)\in r^{\I_{\T,\A}}$. 
                Then, since $\alpha\in\mathbf{C}_y^*$ by definition of set $X$ it follows that $x\not\in X$ and thus $x\neq y$, hence we get from the construction $\mu_1(x)=\mu_2(x), \mu_1(y)=b$. 
                {To sum it up, from the facts that $(a,b)\in r^{\I_{\T,\A}}$, $s$ or $(s^-)$ occurs in the scope of a star in $\pi(x,y)$ with $\transclosure{r}{s}$ and there is a path from $\mu_2(x)$ to $\mu_2(y)$ matching $\pi(x,y)$, we conclude that $\mu_1$ is a valid mapping for $\alpha$.}
            \end{enumerate}
        \item $y$ does not occur in $\alpha$, then we consider the following cases: 
            \begin{enumerate}[label=(\arabic*)]
                \item $\alpha$ has no variable from $X$ and the claim holds by item \ref{proof:clipping:soundness:mapping3} in the construction of mapping $\mu_1$
                \item $\alpha$ is an atom $\bigcup \pi(x,z)$ with $x\in X$ and $\pi$ is either $\node{C}, s$ or $s^*$, which was replaced by $\bigcup \pi(y,z)$ in (C4). By construction of $\mu_1$ we get $\mu_1(x)=a$ and $\mu_1(z)=\mu_1(z)$. Since we have $\bigcup \pi(y,z)$ in $q_2$ it holds that $\mu_1(x)=\mu_2(y)=a$ and $\mu_1(z)=\mu_2(z)$.
                \item The cases for $\bigcup \pi(z,x)$ or $(\bigcup \pi(z,x))^*$ with either $x\in X$ and $z\not\in X$ or ${x,y}\subseteq X$ are analogous to the previous two items. 
            \end{enumerate}
    \end{enumerate}
\end{proof}

\begin{definition}[Degree of mappings]\label{def:degree}
    Given a canonical model $\I_{\T,\A}$ with $\I_0$ being the initial model. 
    Let $|a|=0$ for $a\in\varDelta^{\I_0}$. For $p,w\in\varDelta^{\I_{\T,\A}}$, and $p$ the parent of $w$ in the forest of the anonymous part, let $|w|=|p|+1$. 
    We define the degree of a mapping $\mu$ as $deg(\mu)=\underset{y\in vars(\mu)}{\Sigma |\mu(y)|}$. 
\end{definition}

\begin{lemma}[Completeness Clipping One Step]\label{lemma:clipperonestep:completeness}
    Let $\T$ be an $\ontoLang$ TBox, $\A$ an ABox consistent with $\T$ and $q(\vec{x})$ a \restrictedQuery. If $q_1(\vec{x})$ has a mapping $\mu_1$ in $\A_{\I_{\T,\A}}$ such that $\vec{a}=\mu_1(\vec{x})$ and $deg(\mu_1)>0$, then there exists a $A\ISA \exists r.B\in\T$ and a set $Y$ where $Y\cap\vec{x}=\empty$ such that $q_2(\vec{x})=\mathtt{clipping}(q_1(\vec{x}),A\ISA\exists r.B,Y)$ has a mapping $\mu_2$ in $\A_{\I_{\T,\A}}$ where $\vec{a}=\mu_2(\vec{x})$ and $deg(\mu_2)<deg(\mu_1)$.
\end{lemma}

\begin{proof}
    By assumption $deg(\mu_1)>0$ from which follows that there exists a variable $y$ such that $\mu_1(y)\not\in N$ with $N$ the nodes of $\A$ and $y$ a leaf node (no successors) in the anonymous part. Let $d_y=\mu_1(y)$ and $d_p$ be the parent of $d_y$ that must be introduced by an application of a concept inclusion $A\ISA \exists r.B$ such that $d_p\in A^\I_{\T,\A}$. 
    Then, we choose a set $Y=\{x\in\mathsf{vars}(q_1)\mid \mu_1(x)=d_y\}$ to obtain the query $q_2=\mathtt{clipping}(q_1,A\ISA\exists r.B,Y)$, note since $d_y\not\in N$ it holds that $Y\cap\vec{x}=\emptyset$. 
    
    In the following we show that an atom $\alpha$ in $q_1$ where $y\in Y$ occurs must fall into an item of (C2). Observe that an atom $\alpha$ can have the form (i) $\bigcup C(y)$ (ii) $\bigcup \pi(x,y)$ or (iii) $(\bigcup \pi)^*(x,y)$ for a $\pi:= s\mid s^-\mid \pi^*$ where it either holds for a role $s$ that $\transclosure{r}{s}$ or $r \not\sqsubseteq_\T^* s$. 
    Let's first consider the cases (i) with $\bigcup C'(y)$. Since $\mu_1$ is a mapping for $q_1$, there has to be some concept $C$ in the union, such that we get $\mu_1(y)\in C^{\I_{\T,\A}}$ and by construction of the canonical model it holds that $\mu_1(y)\in B^{\I_{\T,\A}}$. Next, consider the case (ii) with an atom of the form $\bigcup s(x,y)$ such that $\transclosure{r}{s}$. Since $\mu_1$ is a mapping for $q_1$ in $\I_{\T,\A}$, $\mu_1(y)=d_y$ and $d_y$ must not have a self-loop it holds that $x\neq y$. Thus, item (C) applies. The same argument holds for atoms of the form $\bigcup s^-(y,x)$ with $\transclosure{r}{s}$. Given the fact that $d_y$ is a leaf in the anonymous part and $\mu_1(y)=d_y$ with $\mu_1$ a valid mapping for $q_1$ in $\I_{\T,\A}$, there cannot be atoms of the form $\bigcup s(y,x)$ or $\bigcup s^-(x,y)$ in $q_1$. We continue with atoms of the form $\bigcup s^*(x,y)$ or $\bigcup (s^-)^*(y,x)$ with $\transclosure{r}{s}$. If $x=y$, the atom falls into item (A), otherwise into (C). For the inverse directions, i.e., $\bigcup s^*(y,x)$ or $\bigcup (s^-)^*(x,y)$ with $\transclosure{r}{s}$, we have $\mu_1(x)=\mu_1(y)=d_y$ from which follows that $x\in Y$ and thus $x=y$. A mapping of the form $\mu_1(x)\neq \mu_1(y)$ is not valid since $d_y$ is a leaf node and $(d_y,d_p)\not\in s^{\I_{\T,\A}}$. Further, since $\mu_1(y)=d_y$ and $\mu_1$ a valid mapping for $q_1$ over $\I_{\T,\A}$ there can't be atoms containing only roles $s$ without a star and with $r\not\ISA_\T^* s$. In case there exists a role $s$ with a star, but $r \not\ISA_\T^* s$ the only valid mapping must contain $\mu_1(x),\mu_1(y)=d_y$, then $x\in Y$ and the atom falls into (A). The same argumentation holds for case (iii) with $(\bigcup \pi)^*(x,y)$ where the union is under the scope of a star. 
    
    Finally, we construct a mapping $\mu_2$ for $q_2$ in $\A_{\I_{\T,\A}}$ such that $\vec{a}=\mu_2(\vec{x})$ and $deg(\mu_2)\le deg(\mu_1)$ by setting (a) $\mu_2(z)=\mu_1(z)$ for all variables $z$ in $q_2$ with $z\not\in Y$, and (b) $\mu_2(y)=d_p$. It remains to show that $\mu_2$ is a valid mapping for $q_2$ in $\A_{\I_{\T,\A}}$. Observe that $vars(q_2)\subseteq vars(q_1)$ and intuitively $q_2$ is a subquery of $q_1$. Further, by (i) $|\mu_2(z)|=|\mu_1(z)|$ for all $z$ in $q_2$ with $z\not\in Y$ and (ii) $|\mu_2(y)|=|\mu_1(y)|-1$, we conclude that $deg(\mu_2)<deg(\mu_1)$. 
\end{proof}

Using \Cref{lemma:clipperonestep:soundness} and \Cref{lemma:clipperonestep:completeness} that are about one step of the clipping function, we obtain:

\begin{appendixTheorem}{\ref{thm:correctnessNCQRewriting}}[Soundness and Completeness]
	Let $\T$ be an $\ontoLang$ TBox and let $q(\vec{x})$ be a \restrictedQuery. 
    For every ABox $\A$ and tuple $\vec{a}$ of individuals, we have that $\vec{a}$ is a certain answer to $q(\vec{x})$ over $(\T,\A)$ if and only if $\vec{a}$ is an answer to $\mathtt{rewriteNCQ}(q(\vec{x}),\T)$ over $\A$.
\end{appendixTheorem}

\noindent In \Cref{lemma:clipperonestep:soundness} and \Cref{lemma:clipperonestep:completeness} we show that one step of the $\mathtt{clipping}$ function is sound and complete. In \Cref{algo:rewriteC2RPQ} we exhaustively apply this function (\cref{line:clippingLoopStart}-\cref{line:clippingLoopFinish}), which means that we have a query that can be evaluated over the plain ABox without the anonymous part. 
The correctness for the second part of \Cref{algo:rewriteC2RPQ} is given by \Cref{lemma:relation} and exhaustive replacement of all witnessing sets in \cref{line:loopWitnesses}, and application of $\mathtt{rewriteConcept}$ in \cref{line:rewriteConcept} and $\mathtt{rewriteRole}$ in \cref{line:rewriteRole}. \\

\Cref{thm:correctnessNCQRewriting} shows that NCQ answering for $\ontoLang$ can be reduced to C2RPQ query evaluation, which is in \NL in data complexity (see \cite{Barcelo2013}). This is worst-case optimal. 

\begin{appendixTheorem}{\ref{thm:terminationNCQ}}[Termination]
    Let $\T$ be an $\ontoLang$ TBox and $q$ a NCQ. Then, the algorithm $\mathtt{rewriteNCQ}(\T,q)$ terminates. 
\end{appendixTheorem}

\begin{proof} 
    We show termination of $\mathtt{rewriteNCQ}$ for any given $\T$ and $q$ for each of the loops in \Cref{algo:rewriteC2RPQ} separately. 
    Termination of the loop from \cref{line:clippingLoopStart} to \cref{line:clippingLoopFinish} follows from the following facts: 
    \begin{enumerate}
        \item\label{thm:terminationNCQ:item:maxNumAtoms} The maximum number of atoms in the generated query is equal to two times the number of atoms of the initial query. This is due to the fact that atoms satisfying condition (C) of the $\mathtt{clipping}$ function such that $\pi(x,y)\in\textbf{C}_y^*$ might be replaced and not dropped in (C5). However, for each of these atoms, we add at most one additional atom of the form $A(y)$, as this atom fulfils condition (B) in the next iteration and gets dropped in (C3).
        \item\label{thm:terminationNCQ:item:varSetSize} The set of variables that occur in a generated query is a subset of the set of variables $Y$ that occur in $q$. Indeed, in (C1) of the clipping function we replace a subset of $Y$ by a variable $y'\in Y$. 
        \item\label{thm:terminationNCQ:item:numDiffAtoms} From \cref{thm:terminationNCQ:item:varSetSize} it follows that the number of different atoms that may occur in a generated query $q$ is less than or equal to $m*2^n$, where $m$ is the number of axioms in form of $A\ISA \exists r.B\in\T$.  
        \item\label{thm:terminationNCQ:item:dropQueries} We do not drop generated queries.
    \end{enumerate}
    From \cref{thm:terminationNCQ:item:maxNumAtoms} and \cref{thm:terminationNCQ:item:numDiffAtoms} it holds that the number of different queries generated by the algorithm is finite. Further, from \cref{thm:terminationNCQ:item:dropQueries} follows that we do not generate a query more than once, thus it holds that this while loop terminates. 
    
    The termination of the remaining algorithm follows from the fact that the TBox is finite, thus: 
    \begin{enumerate}
        \item $\mathtt{rewriteRole}$ terminates
        \item the \cdg is finite
        \item the number of proper witnessing sets is finite and the function $\mathtt{witnessSets}$ terminates 
        \item the path-generating NFA in \Cref{def:pathGeneratingNFA} produces a finite regular path expression as output of the function $\mathtt{rpe}$ and $\mathtt{rewriteConcept}$ terminates
    \end{enumerate}
\end{proof}

\end{document}